\documentclass[reqno,12pt]{amsart}

\usepackage{geometry,amsmath,amsfonts,amssymb,amsthm,amscd,mathrsfs,graphicx,enumerate,multicol,wrapfig,subfigure,hyperref}
\usepackage[all]{xy}
\usepackage{setspace}

\usepackage[T1]{fontenc}

\numberwithin{equation}{subsection}

\newcommand{\ra}{\rightarrow}
\newcommand{\LRa}{\Longrightarrow}
\newcommand{\lra}{\longrightarrow}

\newcommand{\p}{\prime}
\newcommand{\pt}{\partial}
\newcommand{\eset }{\emptyset}

\newcommand{\al}{\alpha}
\newcommand{\Om}{\Omega}
\newcommand{\om}{\omega}

\newcommand{\Gam}{\Gamma}
\newcommand{\s}{\sigma}
\newcommand{\vp}{\varphi}
\newcommand{\lam}{\lambda}

\newcommand{\q}{\theta}

\newcommand{\dt}{\delta}

\newcommand{\Zbb}{\mathbb{Z}}
\newcommand{\Pbb}{\mathbb{P}}
\newcommand{\Cbb}{\mathbb{C}}

\theoremstyle{plain} 

\newtheorem{THM}{Theorem}[section]
\newtheorem{DEF}[THM]{Definition}
\newtheorem{EX}[THM]{Example}
\newtheorem{CON}[THM]{Construction}
\newtheorem{CONJ}[THM]{Conjecture}
\newtheorem{QUE}[THM]{Question}

\newtheorem{ILL}[THM]{Illustration}
\newtheorem{PROP}[THM]{Proposition}
\newtheorem{LEM}[THM]{Lemma}
\newtheorem{COR}[THM]{Corollary}
\newtheorem{REM}[THM]{Remark}

\newtheorem*{THMSD}{Serre Duality}
\newtheorem*{THMSA}{Serre's Theorem A}

\newcommand{\bt}{\bullet}

\newcommand{\ev}{\mathrm{ev}}
\newcommand{\odd}{\mathrm{odd}}

\newcommand{\img}{\mathrm{im}}

\newcommand{\Vc}{\mathcal{V}}
\newcommand{\Uc}{\mathcal{U}}
\newcommand{\Wc}{\mathcal{W}}
\newcommand{\Oc}{\mathcal{O}}
\newcommand{\red}{\mathrm{red}}

\newcommand{\Der}{\mathrm{Der}}
\newcommand{\Jc}{\mathcal{J}}

\newcommand{\Cc}{\mathcal{C}}
\newcommand{\Gc}{\mathcal{G}}

\newcommand{\Fc}{\mathcal{F}}

\newcommand{\Xfr}{\mathfrak{X}}
\newcommand{\Ec}{\mathcal{E}}

\newcommand{\Ufr}{\mathfrak{U}}

\newcommand{\Tfr}{\mathfrak{T}}
\newcommand{\Mfr}{\mathfrak{M}}

\newcommand{\Mcl}{\mathcal{M}}

\usepackage{xcolor}
\definecolor{airforceblue}{rgb}{0.36, 0.54, 0.66}
\definecolor{burgundy}{rgb}{0.5, 0.0, 0.13}
\definecolor{majorelleblue}{rgb}{0.38, 0.31, 0.86}
\definecolor{darkblue}{rgb}{0.0, 0.0, 0.55}

\hypersetup{urlcolor=burgundy, linkcolor=burgundy,
citecolor=darkblue,
colorlinks=true}

\title{
Obstructed Thickenings and Supermanifolds
\\}
\author{\small Kowshik Bettadapura}
\date{}

\begin{document}
\maketitle

\onehalfspacing

\begin{abstract}
\noindent 
Associated to any supermanifold is a filtration by spaces, referred to as \emph{thickenings}. It is the objective of this article to study them up to a certain equivalence and then up to isomorphism in the complex-analytic setting. We study them from two points of view: (1) as structures embedded in supermanifolds and (2) abstractly. Throughout, we will be guided by the goal to clarify and address the question: \emph{when does a given thickening embed in a supermanifold?} Such a question was, in essence, first studied by Eastwood and LeBrun. In this article we begin with a pedagogical account of their study, after which we further study thickenings in supergeometry and present a classification of thickenings of a given order.  As a complement to our study, we comment on the moduli problem for complex supermanifolds and consider the analogous problem for thickenings. Finally, to illustrate the ideas in this article, we conclude by describing some obstructed thickenings of the complex projective plane. 
\end{abstract}

\setcounter{tocdepth}{1}
\tableofcontents

\section{Introduction}

\noindent
Supermanifolds are objects which have arisen from considerations of supersymmetry in physics. In the context of superstring theory, super Riemann surfaces and their moduli are of particular interest, for it is on the moduli space of super Riemann surfaces where one calculates scattering amplitudes. The methods for undertaking these calculations are sensitive to the `obstruction theory' of the moduli space, which was the subject of study by Donagi and Witten in \cite{DW1, DW2}. Hence, for at least this reason, one can find motivation for studying the obstruction theory of supermanifolds more generally. 
\\

\noindent
There are a number articles in the supergeometry literature which address obstruction theory specifically. We reference a few here. In \cite{ONISHNS, ONISHNSCOT} one finds explicit constructions of complex supermanifolds. Obstruction theory here is studied by reference to obstruction classes. The notion of `higher obstruction classes', is introduced in \cite{BER, DW1} and studied further in \cite{BETTHIGHOBS}. Obstruction theory related to embeddings of supermanifolds are studied in \cite{LEBRUN, BETTEMB, NOJAEMB}; and for varieties in projective superspace in \cite{BETTVAR}. For other recent and exciting developments in complex supergeometry related to obstruction theory, see \cite{NOJACY, NOJAONESCY, NOJAPI}.   
\\

\noindent
In this article we study obstruction theory with the view to further understand the classification of supermanifolds. This problem, of classifying supermanifolds, was itself initiated by Batchelor in \cite{BAT} where a classification is obtained in the smooth setting. In the complex-analytic setting, the analogous problem becomes more subtle and inroads into the classification problem are made in \cite{GREEN, YMAN, ONISHCLASS}. The moduli variety of complex supermanifolds was constructed and studied by Onishchik in a particular instance in \cite{ONISHMOD} and then in more generality in \cite{ONISHCLASS}. We review these studies here and, in line with the overall theme of this article, propose that this variety will admit a filtration. 
\\

\noindent The primary focus of this article however is on studying the obstruction theory of a more fundamental construct than that of a supermanifold, being that of a \emph{thickening}. These thickenings are more fundamental in the sense that: \emph{associated to any supermanifold will be a thickening, but not necessarily conversely}. In taking inspiration from the study of thickenings in complex geometry by Griffiths in \cite{GRIFF},  analogous notions are developed by Eastwood and LeBrun in \cite{EASTBRU}. There, the rudiments of an obstruction theory incorporating thickenings is laid out and it is this obstruction theory that we investigate further in this article. We begin with a pedagogical review of the work in \cite{EASTBRU} by reference to methods similar to those of Kodaira-Spencer deformation theory, expounded in \cite{KS}. To justify this method, we re-derive one of the main results in \cite{EASTBRU} pertaining to the identification of the space of obstructions to finding thickenings,  leading then to a classification.
\\\\
Regarding the classification problem for supermanifolds, the idea here behind studying thickenings is the following: in order to classify supermanifolds, one might instead try and classify these thickenings and throw out those which will not `give rise' to a supermanifold. Along this vein, our first step is in observing that a thickening will come in two basic flavours: \emph{unobstructed} and \emph{obstructed}. The obstructed thickenings will never give rise to a supermanifold whereas the unobstructed thickenings \emph{may} or \emph{may not}. In this way we can relate the obstruction theory for supermanifolds with that for thickenings. These observations are formulated at the level of cohomology, leading to meaningful decompositions of certain cohomology groups which we describe here. 
\\\\
The main new result in this article is in the construction of obstructed thickenings of the complex projective plane. We infer the existence of obstructed thickenings of $\Cbb\Pbb^2$ equipped with a rank 3, holomorphic vector bundle $E$ which is: (1) split; and (2) non-split but decomposable.

\subsection{Outline and Summary}
This article is divided into eight sections, including this introduction and some concluding remarks. The contents of the other sections are briefly summarised below.
\\\\
In Section \ref{2} the notation is set and some relevant, background theory is provided so as to be called upon as needed throughout this article. In Section \ref{3} we discuss thickenings as developed by Eastwood and LeBrun in \cite{EASTBRU} and set the motivating theme, guiding considerations in the sections to come. Importantly, the notion of an \emph{obstructed thickening} is defined and one of the central results in \cite{EASTBRU} is stated here in Theorem \ref{jnvjknknvkjnvk}, pertaining to the existence of obstructed thickenings. 

Section \ref{4} is devoted to a proof of this result (Theorem \ref{jnvjknknvkjnvk}) by methods similar to those found in Kodaira-Spencer deformation theory, detailed in \cite{KS}. This involves describing a construction of thickenings by means of gluing and working directly with this gluing data. As a prelude, to justify our argument, we provide a `direct' proof of Theorem \ref{jnvjknknvkjnvk} in a particular instance. To then give some further context in which this theorem can sit, we look at the problem of classifying thickenings of a given order, culminating in Theorem \ref{o4993uf0oijwoijfwojpw}.

In Section \ref{5}, we recall the second object in title of this article: \emph{supermanifolds}. After a brief interlude on obstruction theory for supermanifolds, we look to integrate this obstruction theory with the more general theory for thickenings. This is expressed in the form of a decomoposition of certain 1-cohomology groups valued in appropriate sheaves of abelian groups. A precise formulation is given in \eqref{49d93j8fj3fj3p}.

We have so far been studying thickenings up to a certain equivalence that is stronger than isomorphism. In Section \ref 6 we explore moduli problems, which considers these objects---supermanifolds and thickenings, up to isomorphism. The moduli problem for complex supermanifolds was studied by Onishchik in \cite{ONISHCLASS} from an analytic point of view and by Vaintrob in \cite{VAIN} from an algebraic point of view. We propose here a relation between these viewpoints in Conjecture \ref{fkrpokvpokvp4k}. We then consider the moduli variety constructed by Onishchik in \cite{ONISHCLASS} and propose that it admit a filtration in analogy with supermanifolds in Conjecture \ref{dnohf3fh839hf3h}. We conclude by formulating the moduli problem for thickenings and argue that the previous-mentioned decomposition in \eqref{49d93j8fj3fj3p} will hold here, i.e., up to isomorphism. 

In the penultimate section, Section \ref{7}, we present the main new result in this article, being the construction of obstructed thickenings of the complex projective plane $\Cbb\Pbb^2$. We consider the case where $\Cbb\Pbb^2$ is equipped with a rank 3, holomorphic vector bundle which is either: (1) split; and (2) non-split but decomposable.

\subsection*{Acknowledgements}
I would like to acknowledge the many helpful discussions and support I have had with Peter Bouwknegt and Bryan Wang; the Australian Postgraduate Award for providing financial support during the time when much of the work in this article was undertaken; and the useful comments of the anonymous referee.

\section{Preliminaries}\label{2}

\subsection{Supermanifolds}

The definition of a supermanifold may be quite succinctly given in the framework of algebraic geometry. We refer to \cite{QFAS} where this point of view is emphasised. We begin firstly with the following:

\begin{DEF}\label{mironufuir}
\emph{A complex manifold $M$ of dimension $n$ is a locally ringed space $(|M|, \Cc_M)$, for $|M|$ a topological space, which is locally isomorphic to $(\Cbb^n, \Cc_{\Cbb^n})$, where $\Cc_{\Cbb^n}$ denotes the sheaf of (germs of) holomorphic functions on $\Cbb^n$. 
}
\end{DEF}

The isomorphism mentioned in the definition of a complex manifold in Definition \ref{mironufuir} above is in the category of \emph{locally ringed spaces}. Note that a (real) smooth manifold may be defined analogously. The definition of \emph{supermanifold} is now somewhat natural:

\begin{DEF}\label{ncnrnv322rnk}
\emph{A $(p|q)$-dimensional real (resp. complex) supermanifold $\Xfr$ is defined as a locally ringed space\footnote{\label{poirrjnkeoeo}A technical point, which we finesse here, is the distinction between commutative rings and their $\Zbb_2$-graded counterparts, the so-called super-commutative rings. A locally ringed space typically requires the structure sheaf to be a sheaf of commutative rings. We need to consider here sheaves of super-commutative rings however. That the notion of a locally ringed space in this context also makes sense is discussed in \cite{LEI,BER, YMAN, KAPSUSY}. In light of this, we will continue on with the terminology `locally ringed space'.
} $(M, \Oc_M)$, where $M$ is a $p$-dimensional, real (resp. complex) manifold and the structure sheaf $\Oc_M$ is:
\begin{enumerate}[(i)]
	\item  $\Zbb_2$-graded;
	\item equipped with an epimorphism $\iota^\sharp : \Oc_M \ra \Cc_M$ with kernel $\Jc = \ker\iota^\sharp$ and;
	\item locally isomorphic to the sheaf of algebras $\Cc_M\otimes \wedge^\bt V$, where $V$ is a fixed, real (resp. complex) vector space of dimension $q$.
\end{enumerate}
}
\end{DEF}

In what follows we make sense of describing a supermanifold as being `modelled' on a manifold $M$ and vector bundle $E\ra M$.

\subsection{The Split Model}\label{fjf9j490fj40fj40}
A supermanifold $\Xfr$ of dimension $(p|q)$ is said to have \emph{even} dimension $p$ and \emph{odd} dimension $q$. For the purposes of this article it will be convenient to think about a supermanifold $\Xfr$ as being modelled on two bits of data: a manifold $M$, called the \emph{reduced space}, and a vector bundle $E\ra M$ called the \emph{modelling bundle}. Let $\Ec$ denote the sheaf of sections of $E$. To see how to make sense of a supermanifold $\Xfr$ as being `modelled' on $(M, E)$, recall from Definition \ref{ncnrnv322rnk}(ii) the epimorphism $\iota^\sharp : \Oc_M \ra \Cc_M$ and set $\Jc = \ker\iota^\sharp$. It is called the \emph{nilpotent ideal}. By Definition \ref{ncnrnv322rnk}(iii) we see that the quotient $\Jc/\Jc^2$ will be a sheaf of locally free $\Cc_M$-modules and hence correspond to the sheaf of sections of a vector bundle, say $E\ra M$. The structure sheaf $\Oc_M$ of $\Xfr = (M, \Oc_M)$ will then be locally isomorphic to the sheaf of sections $\wedge^\bt \Ec$ of the bundle of exterior algebras $\wedge^\bt E$. If there exists an isomorphism $\Jc/\Jc^2\stackrel{\cong}{\ra}\Ec$, then we say $\Xfr$ will be \emph{modelled} on a pair $(M, E)$ and use the notation $\Xfr_{(M, E)}$.

\begin{REM}
\emph{The distinction between $\Xfr$ and $\Xfr_{(M,E)}$ is superficial. We make it only to emphasise that, in this article, we adopt a `bottom-up' point of view on supermanifolds. That is, rather than starting with some supermanifold, we start with a manifold $M$, a vector bundle $E\ra M$ and study supermanifold structures associated to $(M, E)$.}
\end{REM}

\begin{DEF}\label{ifh4fgiuho48h04}
\emph{For a supermanifold $\Xfr_{(M, E)}$, a choice of isomorphism $\vp : \Jc/\Jc^2\stackrel{\cong}{\ra}\Ec$ is referred to as a \emph{coframe} for $\Xfr_{(M, E)}$. 
}
\end{DEF}

We will justify the terminology `coframe' in the above definition later on in this article. 

\begin{DEF}
\emph{The supermanifold $\Xfr_{(M, E)}$ is said to be \emph{smooth} (resp. \emph{holomorphic}) if the pair $(M, E)$ consists of a smooth (resp. complex-analytic) manifold and smooth (resp. holomorphic) vector bundle.
}
\end{DEF}

In this article we will be concerned with holomorphic supermanifolds.
\\\\
Now from any given pair $(M, E)$ we may construct a supermanifold by simply taking the structure sheaf $\Oc_M$ to be $\wedge^\bt \Ec$. This gives what is termed the \emph{split model} and is denoted here by $\Pi E$. As a locally ringed space $\Pi E =  (M, \wedge^\bt\Ec)$. A supermanifold $\Xfr$ modelled on a pair $(M, E)$ is locally isomorphic to $\Pi E$. To see why, note by Definition \ref{ncnrnv322rnk}(iii) that $\Xfr$ will be locally isomorphic to $\wedge^\bt(\Jc/\Jc^2)$. If $\Xfr$ is modelled on $(M, E)$, then $\Jc/\Jc^2\cong \Ec$. Therefore $\wedge^\bt(\Jc/\Jc^2)\cong \wedge^\bt\Ec$ and so $\Xfr$ and $\Pi E$ are locally isomorphic.

\begin{REM}
\emph{As $\Oc_M$ is $\Zbb_2$-graded then, if $\Oc^{\ev/\odd}_M$ denote the even and odd components respectively, we have: $\Oc^{\ev/\odd}_M \cong_{\mathrm{loc.}}\wedge^{\ev/\odd}\Ec$ as sheaves of $\Cc_M$-modules.
}
\end{REM}

\subsection{Embedded Thickenings}\label{kd4k9fk30fk30kf3p}
It is our intent in this article to study thickenings independently of supermanifolds, in a suitable sense. Before giving their definition however, we will present the following construction associated to any supermanifold. It can be found in \cite{BER, YMAN}.

\begin{CON}\label{nruiciurhf94h9h}
Let $\Xfr = (M, \Oc_M)$ be a supermanifold and denote by $\Jc$ the nilpotent ideal. Recall that it is identified with the kernel of the morphism $\Oc_M \twoheadrightarrow \Cc_M$. That is, we have an exact sequence of sheaves of $\Cc_M$-modules:
\[
0 \ra \Jc \hookrightarrow \Oc_M \twoheadrightarrow \Cc_M\ra0.
\]
Consider now the $\Jc$-adic filtration on $\Oc_M$, i.e., $\Oc_M\supset \Jc\supset \Jc^2\supset\cdots$. If $\Xfr$ has odd-dimension $q$, then the $\Jc$-adic filtration will have length $q$, i.e., $\Jc^{q+1} = 0$. Now set $\Oc_M^{(k)} := \Oc_M/\Jc^{k+1}$. Then we will obtain, in this way, a locally ringed space $\Xfr^{(k)} = (M, \Oc^{(k)}_M)$, with: $\Xfr^{(k)}\subset \Xfr$. 
\end{CON}

We now have the following remarks: if $k =0$, then we recover the underlying, complex manifold since $\Oc_M^{(0)} = \Oc_M/\Jc = \Cc_M$. If we fix a pair $(M, E)$, then $\Xfr^{(1)}_{(M, E)}$ and $\Pi E^{(1)}$ coincide. Finally, if $k \geq q$ then clearly $\Oc_M^{(k)} = \Oc_M$ and so $\Xfr^{(k)} = \Xfr$. Note that the sheaf $\Oc_M^{(k)}$, for $0 < k < q$, will not be locally isomorphic to a sheaf of exterior algebras, and so the thickenings $\Xfr^{(k)}$ will not be supermanifolds in the sense of Definition \ref{ncnrnv322rnk}. Regarding $\Xfr_{(M,E)}$, we have by construction a filtration:
\begin{align}
M \subset \Pi E^{(1)} \subset \Xfr_{(M, E)}^{(2)}\subset \cdots\subset \Xfr^{(q-1)}_{(M, E)} \subset \Xfr_{(M, E)}. 
\label{bcbfjhbvrbjh}
\end{align}
In light of \eqref{bcbfjhbvrbjh} we have the following definition, motivated by \cite{EASTBRU}:

\begin{DEF}\label{nkcdkjkfoieoe}
\emph{The locally ringed space $\Xfr_{(M, E)}^{(k)}$ in \eqref{bcbfjhbvrbjh} is referred to as the \emph{$k$-th order thickening of $M$ in $\Xfr_{(M, E)}$}.
}
\end{DEF}

Another way to think of a $k$-th order thickening $\Xfr^{(k)}_{(M, E)}$ in $\Xfr_{(M, E)}$ is as a thickening of $M$ equipped with an embedding $\iota: \Xfr^{(k)}_{(M, E)} \subset\Xfr_{(M, E)}$.\footnote{\label{rfb8y4f84hf9h8f3}As described in \cite{LEBRUN} and motivated by notions in algebraic geometry, an embedding of supermanifolds (and thickenings more generally) is defined by: (1) an embedding of underlying spaces and; (2) a surjection of structure sheaves. For another description of embeddings, see \cite{BETTEMB}. In the present case, $\iota: \Xfr^{(k)}_{(M, E)} \ra\Xfr_{(M, E)}$ is the identity on reduced spaces and surjective on structure sheaves by construction. Hence it is an embedding.} Hence we may refer to such a thickening as an \emph{embedded} thickening. As mentioned at the start of this section, we are interested in thickenings `independently' of supermanifolds and by this we mean thickenings that need not be embedded in some higher structure such as a supermanifold, i.e., an \emph{abstract} thickening. The focus of this article is then: \emph{given an abstract thickening, will there exist an embedding of it into some higher structure, e.g., such as a thickening of higher order, or a supermanifold?}
\\\\
 In the next section we will look at the tangent sheaf of a supermanifold and the split model.

\subsection{The Tangent Sheaf}
Let $\Oc_U$ be a sheaf of (super-)commutative rings on a topological space $U$. Then the sheaf of derivations $\Der~\Oc_U$ is locally free, as shown in \cite{LEI, QFAS}. The \emph{tangent sheaf} of a supermanifold $\Xfr = (M, \Oc_M)$ is then defined to be $\Tfr_\Xfr :=\Der~\Oc_M$. It is locally free and identified with the sheaf of sections of the \emph{tangent bundle} of $\Xfr$. We will work directly with the sheaf $\Tfr_\Xfr$ rather than the bundle in this article. Our first observation here is: \emph{since the split model $\Pi E$ is $\Zbb$-graded then so is the tangent sheaf $\Tfr_{\Pi E}$}. As a supermanifold $\Xfr$ is more generally only $\Zbb_2$-graded, it follows that the tangent sheaf $\Tfr_\Xfr$ is $\Zbb_2$-graded. It will admit more structure than simply a $\Zbb_2$-grading however. As described in \cite{ONISHCLASS}, it is in fact a filtered $\Oc_M$-module. 

\begin{CON}
Let $\Xfr$ be a $(p|q)$-dimensional supermanifold and set:
\[
\Tfr_\Xfr[k] := \left\{ \nu\in \Der~\Oc_M \mid \nu(\Oc_M)\subset \Jc^k~~\mbox{and}~~\nu(\Jc^l)\subset \Jc^{k+l}~\mbox{for all $l>0$}\right\}.
\]
Then since $\Jc\supset \Jc^2\supset\cdots$, it follows that $\Tfr_\Xfr[k]\supset \Tfr_\Xfr[k+1]$ for all $k$. Set $\Tfr_\Xfr[-1] := \Tfr_\Xfr$. We now have:
\[
\Tfr_\Xfr = \Tfr_\Xfr[-1] \supset \Tfr_\Xfr[0] \supset \Tfr_\Xfr[1]\supset\cdots\supset \Tfr_\Xfr[q]\supset \Tfr_\Xfr[q+1]=0.
\]
In this way $\Tfr_\Xfr$ is a filtered $\Oc_M$-module. 
\end{CON}

Now let $\Xfr$ be a supermanifold modelled on $(M, E)$. Then it is locally isomorphic to its split model $\Pi E$. Denote by $\Tfr_{(M, E)}$ its tangent sheaf. We have the following result, an argument for which we refer to \cite[p. 311]{ONISHCLASS}:

\begin{LEM}\label{jd90d098jdoijp3}
There exists a short-exact sequence of sheaves of $\Cc_M$-modules: 
\[
0 \ra \Tfr_{(M, E)}[k+1] \hookrightarrow \Tfr_{(M, E)}[k] \twoheadrightarrow \Tfr_{\Pi E}[k]\ra0
\]
for all $k\geq-1$.
\qed
\end{LEM}

As for the tangent sheaf of the split model itself, we have the following:

\begin{LEM}\label{hjbbjhhveyjj}
There exists a short-exact sequence of sheaves of $\Cc_M$-modules: 
\begin{align}
0 \ra \wedge^{k+1}\Ec\otimes\Ec^\vee \hookrightarrow  \Tfr_{\Pi E}[k] \twoheadrightarrow \Tfr_M\otimes\wedge^k\Ec\ra0
\label{jkkbkbkbk}
\end{align}
for all $k\geq-1$.
\qed
\end{LEM}

\section{Thickenings in Supergeometry}\label{3}

\noindent
In Construction \ref{nruiciurhf94h9h} we obtained locally ringed spaces from a given supermanifold $\Xfr$. To indicate this dependence on $\Xfr$, we referred to these spaces as \emph{thickenings in $\Xfr$} in Definition \ref{nkcdkjkfoieoe}. In the present section we will give a treatment of thickenings in the spirit of \cite{EASTBRU}.

\subsection{Preliminaries: Thickenings of Complex Manifolds}
We refer to \cite{EASTBRU,GRIFF} for a more complete treatment on the theory of thickenings of complex manifolds. Here only the basic notions will be described with the goal to motivate similar considerations in the context of supergeometry. 
\\\\
Firstly fix a complex manifold $M = (|M|, \Cc_M)$, thought of here as a locally ringed space (see Definition \ref{mironufuir}). 

\begin{DEF}
\label{djcnknckjncjkrccr}
\emph{A thickening of $M$ of order $m$ is defined to be a locally ringed space $M^{(m)} = (M, \Oc_M^{(m)})$, where the structure sheaf $\Oc_M^{(m)}$ is equipped with: 
\begin{enumerate}[(i)]
	\item an epimorphism $\Oc^{(m)}_M \twoheadrightarrow \Cc_M$; and
	\item local isomorphisms: $\Oc_M^{(m)} \cong_{\mathrm{loc.}} \Cc_M[{\bf x}]/({\bf x}^{m+1})$ 
\end{enumerate}
for ${\bf x}$ a formal variable.
}
\end{DEF}

An instructive example of a thickening is provided by holomorphically embedded sub-manifolds:

\begin{EX}
\label{jcnkrncjkrnckjr}
Suppose $N$ is a complex manifold and $M\subset N$ is a holomorphically embedded, co-dimension one submanifold.  If $\Cc_N$ denotes the structure sheaf of $N$, consider the ideal $\Jc_M\subset \Cc_N$ comprising those functions on $N$ which vanish on $M$. Then an $m$-th order thickening of $M$ is given by the locally ringed space $M^{(m)} = (M, \Oc_M^{(m)})$, where $\Oc_M^{(m)} := \iota^*(\Cc_N/\Jc_M^{m+1})$, for $\iota : M\hookrightarrow N$ the embedding. 
\end{EX}

Indeed, given a thickening of order $m$, it is possible to construct another thickening of order $l<m$ in the same manner as in Construction \ref{nruiciurhf94h9h}.

\begin{CON}
\label{djcnknjcknjcknrkc}
Suppose we are given an $m$-th order thickening $M^{(m)}$ of $M$. Then we obtain an ideal $\Jc_{(m)}$ as the following kernel,
\begin{align}
\Jc_{(m)} := \ker\{ \Oc_M^{(m)} \twoheadrightarrow \Cc_M\}. 
\label{jcnkrnckrnckjrcrcr}
\end{align}
Now consider the $\Jc_{(m)}$-adic filtration of $\Oc_M^{(m)}$ and define $\Oc_M^{(l)}$ by,
\begin{align}
\Jc^{l+1}_{(m)} \hookrightarrow \Oc_M^{(m)} \twoheadrightarrow \Oc_M^{(l)}
\label{jcnkjrnckrcnkrnc}
\end{align}
for $0\leq l< m$. Then $(|M|, \Oc_M^{(l)})$ will be a thickening of $M$ of order $l$.
\end{CON}

Now by definition of $\Jc_{(m)}$ in \eqref{jcnkrnckrnckjrcrcr} we see that $M = M^{(0)}$. Moreover, since $\Jc^{l}_{(m)}\supset \Jc^{l+1}_{(m)}$, we are led to the following commuting diagram of exact sequences:
\[
\xymatrix{
 0 \ar[r] & \ar@{^{(}->}[d] \Jc^{l+1}_{(m)} \ar[r] & \ar@{=}[d]\Oc_M^{(m)} \ar[r] & \ar[d]\Oc_M^{(l)} \ar[r] & 0\\
0 \ar[r] & \Jc^{l}_{(m)} \ar[r] & \Oc_M^{(m)} \ar[r] & \Oc_M^{(l-1)} \ar[r] & 0
}
\]
The morphism $\Oc_M^{(l)} \ra \Oc_M^{(l-1)}$ will be surjective. As these are the structure sheaves of the ringed spaces $M^{(l)}$ and $M^{(l-1)}$, this morphism will therefore correspond to an embedding $M^{(l-1)}\subset M^{(l)}$ (c.f., footnote \eqref{rfb8y4f84hf9h8f3}). In general we have a filtration
\begin{align}
M = M^{(0)} \subset M^{(1)} \subset M^{(2)} \subset \cdots\subset M^{(m)}.
\label{idcmlmclrmclr}
\end{align}
In this way we see that associated to any $m$-th order thickening of $M$ are thickenings of orders $l$, for $l = 0, \ldots, m-1$, filtered as in \eqref{idcmlmclrmclr}. We now have the following  natural question: 

\begin{QUE}\label{thickeningsquestion}
Given a thickening $M^{(m)}$ of $M$ does there exist a thickening $M^{(m+1)}$ of $M$ containing $M^{(m)}$?
\end{QUE}

One of the main results in \cite{EASTBRU} is in identifying the space of obstructions to the existence of thickenings $M^{(m+1)}$ of a given $M^{(m)}$. This was then generalised to the context of supergeometry. This is what we detail in what follows.

\subsection{Thickenings in Supergeometry}
We firstly note the similarities between Definition \ref{djcnknckjncjkrccr} of a thickening of a complex manifold and Definition \ref{ncnrnv322rnk} of a (complex) supermanifold. This motivates the following definition of an abstract, order-$m$ `super-geometric' thickening of a complex manifold:

\begin{DEF}\label{xomcomocoocoko}
\emph{Let $M$ be a $p$-dimensional, complex manifold. An \emph{abstract, super-geometric thickening of $M$ of order $m$ and dimension $(p|q)$} is a locally ringed space $\Xfr^{(m)} = (M, \Oc_M^{(m)})$ whose structure sheaf satisfies:
\begin{enumerate}[(i)]
	\item it is $\Zbb_2$-graded;
	\item it is equipped with an epimorphism $\iota^\sharp: \Oc_M^{(m)}\twoheadrightarrow \Cc_M$; and 
	\item it is locally isomorphic to $\Cc_M[\xi_1, \ldots, \xi_q]/J^{m+1}$, where $\xi_i$ anti-commute amongst each other; $J = \ker\{\Cc_M[\xi_1, \ldots, \xi_q]\twoheadrightarrow \Cc_M\}$; and $J^{m+1}$ is its $(m+1)$-th power.
\end{enumerate}
}
\end{DEF}

\begin{REM}\label{rfi4hf984hoj3jp3}\emph{The definition of an abstract, super-geometric thickening given in Definition \ref{xomcomocoocoko} above is taken from \cite{EASTBRU} where the term `\emph{supersymmetric} thickening' is used. The word supersymmetry has some connotation in physics so we have opted here to use the term `super-geometric' rather than supersymmetric. However, we will later drop the prefix `super-geometric' where confusion is unlikely to arise.}
\end{REM}

\begin{DEF}\label{deduieuidh4hd487h}
\emph{Let $\Xfr^{(m)}$ be an abstract, super-geometric thickening of $M$ of order-$m$. We say it is \emph{trivial} if the local isomorphism in Definition \ref{xomcomocoocoko}(iii) is global.}
\end{DEF}

As with supermanifolds, if $\Jc = \ker\iota^\sharp$ then $\Jc/\Jc^2$ is identified with the sheaf of sections of a vector bundle $\Ec$ of $E\ra M$. So in analogy with the notion of an abstract supermanifold modelled on a pair $(M, E)$, we have the notion of an abstract, \emph{super-geometric thickening} of the pair $(M, E)$ of order $m$, denoted $\Xfr^{(m)}_{(M, E)}$. It is defined just as in Definition \ref{xomcomocoocoko} with Condition (iii) replaced by:
\begin{align}
\Oc_M^{(m)}\cong_{\mathrm{loc.}} \wedge^{\bt}\Ec/J^{m+1},
\label{nfo3hof3hf8o3hf03}
\end{align}
where $J= \ker\{\wedge^\bt\Ec \twoheadrightarrow \Cc_M\}$.

\begin{LEM}\label{fiugf78g79h380fj09}
For any thickening $\Xfr^{(m)}_{(M, E)}$ we have $\Xfr^{(1)}_{(M, E)} \cong \Pi E^{(1)}$.  
\end{LEM}

\begin{proof}
This follows from Definition \ref{xomcomocoocoko}(i). To see this, recall that $\Xfr^{(m)}_{(M, E)}$ is the thickening $\Xfr^{(m)} = (M, \Oc^{(m)}_M)$. Moreover, we have an isomorphism $\Jc/\Jc^2\stackrel{\cong}{\ra}\Ec$. Now since $\Oc^{(m)}_M$ is $\Zbb_2$-graded we can write $\Oc^{(m)}_M = \Oc_M^{(m);\ev}\oplus \Oc_M^{(m);\odd}$. The ideal $\Jc$ is generated by odd sections and so 
\begin{align}
\Oc_M^{(m);\ev}/\Jc^{2k} = \Oc_M^{(m);\ev}/\Jc^{2k+1}
&&
\mbox{and}
&&
\Oc_M^{(m);\odd}/\Jc^{2k+1} = \Oc_M^{(m);\odd}/\Jc^{2k+2}.
\label{3hf8gf83gf793hf39}
\end{align}
In particular, for $\Xfr^{(1)}_{(M, E)}$ we find:
\begin{align*}
\Oc_M^{(1)} := \Oc^{(m)}_M/\Jc^2
&= (\Oc_M^{(m);\ev}/\Jc^2)\oplus(\Oc_M^{(m);\odd}/\Jc^2)
\\
&= (\Oc^{(m)}_M/\Jc)\oplus(\Jc/\Jc^2) &&\mbox{from \eqref{3hf8gf83gf793hf39}}
\\
&\cong
\Cc_M\oplus \Ec \\
&= \wedge^\bt\Ec^{(1)}.
\end{align*}
Hence $\Xfr^{(1)}_{(M, E)}\cong \Pi E^{(1)}$. 
\end{proof}

\begin{REM}\label{dnclncjkdddjdjd}
\emph{For $E\ra M$ a vector bundle of rank $q$, a thickening of $(M, E)$ of order $q$ or higher will define a $(p|q)$-dimensional supermanifold $\Xfr$ modelled on $(M, E)$. In this way, a supermanifold $\Xfr_{(M, E)}$ may be thought of as a `maximal', super-geometric thickening of the pair $(M, E)$. 
}
\end{REM}

\begin{REM}\emph{To avoid cumbersome terminology, we will typically refer to a `$(p|q)$-dimensional, order $m$, abstract, super-geometric thickening' as a \emph{$(p|q)$-dimensional thickening}. Moreover, if $E\ra M$ is a fixed, rank-$q$ vector bundle then the terminology `$(p|q)$-dimensional' is redundant and so will be dropped.} 
\end{REM}

Now similarly to \eqref{bcbfjhbvrbjh} and \eqref{idcmlmclrmclr}, we have a filtration of locally ringed spaces:
\begin{align}
M \subset \Pi E^{(1)} \subset \Xfr_{(M, E)}^{(2)}\subset\cdots\subset\Xfr_{(M, E)}^{(m)}.
\label{rhcbjrbchjrchjrnc}
\end{align}
Then just as in the case of thickenings of complex manifolds, we want to know what the obstructions are to extending a given, abstract thickening. That is, we want a definitive answer to Question \ref{thickeningsquestion}, adapted to thickenings in the present context: 

\begin{QUE}\label{SUSYthickeningsquestion}
Given an abstract (super-geometric) thickening $\Xfr_{(M, E)}^{(m)}$ of $M$, does there exist a (super-geometric) thickening $\Xfr_{(M, E)}^{(m+1)}$ which contains $\Xfr_{(M, E)}^{(m)}$?
\end{QUE}

In order to gain traction on the Question \ref{SUSYthickeningsquestion} we submit the following definition.

\begin{DEF}\label{smeioiof894f89jcoi}
\emph{A given, abstract thickening $\Xfr_{(M, E)}^{(m)}$ is said to be \emph{obstructed} if there does not exist any $(m+1)$-th order thickening $\Xfr_{(M, E)}^{(m+1)}$ containing it. Otherwise, it is said to be \emph{unobstructed}.
}
\end{DEF}

Hence if a given, abstract thickening $\Xfr_{(M, E)}^{(m)}$ is unobstructed, there will exist a thickening $\Xfr_{(M, E)}^{(m+1)}$ containing it and so an embedding $\Xfr_{(M, E)}^{(m)}\subset \Xfr_{(M, E)}^{(m+1)}$. In this article one of the goals is to give an example of an obstructed thickening, i.e., an abstract thickening for which no embedding into a higher structure will exist. In the sections to follow we will look to give a classification of thickenings and a first step toward doing so is a notion of equivalence for thickenings, defined below.

\begin{DEF}\label{equivextensusythick}
\emph{Two thickenings $\Xfr_{(M, E)}^{(m+1)}$ and $\tilde\Xfr_{(M, E)}^{(m+1)}$ of $\Xfr^{(m)}$ are said to be \emph{equivalent} if there exists an isomorphism $\Xfr_{(M, E)}^{(m+1)} \stackrel{\cong}{\ra}\tilde \Xfr_{(M, E)}^{(m+1)}$, as locally ringed spaces, which restricts to the identity on $\Xfr_{(M, E)}^{(m)}$, i.e., a commutative diagram of locally ringed spaces:
\[
\xymatrix{
\ar@{^{(}->}[d] \Xfr_{(M, E)}^{(m)} \ar@{^{(}->}[r] & \Xfr_{(M, E)}^{(m+1)}\ar[dl]_\cong\\
\tilde \Xfr_{(M, E)}^{(m+1)}. & 
}
\] 
We denote by ${\bf T}^1(\Xfr_{(M, E)}^{(m)})$ the set of equivalence classes of $(m+1)$-th order thickenings containing $\Xfr_{(M, E)}^{(m)}$.
}
\end{DEF}

A first step toward an answer to Question \ref{SUSYthickeningsquestion} is provided by the following result:

\begin{THM}
\label{jnvjknknvkjnvk}
Suppose $\Xfr_{(M, E)}^{(m)}$ is a thickening of $(M,E)$ of order $m$, where $m\geq2$. Then the space of obstructions to finding a thickening of $\Xfr_{(M, E)}^{(m)}$ is identified with either: 
\begin{align*}
\mbox{$H^2(M, \Tfr_M\otimes\wedge^{m+1}\Ec)$ if $m$ is odd}
&&\mbox{or}&&
\mbox{$H^2(M, \wedge^{m+1}\Ec\otimes \Ec^\vee)$ if $m$ is even.} 
\end{align*}
\end{THM}
~\\

Theorem \ref{jnvjknknvkjnvk} is stated in \cite[p. 1188]{EASTBRU}, where it is deduced from the more general theory of extensions of holomorphic vector bundles, detailed in \cite{GRIFF}. We will give a more elementary proof of Theorem \ref{jnvjknknvkjnvk} in the section to follow.

\section{Obstructions to Existence}\label{4}

\noindent
In Remark \ref{dnclncjkdddjdjd} it was observed that a supermanifold modelled on $(M, E)$ may be thought of as a `maximal' super-geometric thickening of $(M, E)$. Question \ref{SUSYthickeningsquestion} then motivates asking: \emph{given a thickening $\Xfr^{(m)}_{(M, E)}$, does there exist a supermanifold $\Xfr_{(M, E)}$ which contains it?} We will ask this question again later. For now, we concentrate on Theorem \ref{jnvjknknvkjnvk} where the space of obstructions to the existence of such a supermanifold is identified. We will present here firstly a `direct' proof of Theorem \ref{jnvjknknvkjnvk} for $m = 2$, from which a useful criterion for checking whether a given thickening is obstructed or not will become apparent. A full (and more indirect) proof of Theorem \ref{jnvjknknvkjnvk} will then be given. To give these proofs, it will be convenient to digress and describe thickenings by means of gluing.

\subsection{Gluing Data for Thickenings} Associated to any manifold (smooth or complex) is the data of an atlas, charts and transition functions. Conversely, the \emph{manifold gluing construction} allows one to construct a manifold from a given collection of such data. This is described in \cite{LEE} for smooth manifolds and in \cite{KS} for complex manifolds. Importantly, a similar description exists for supermanifolds and thickenings. 

\begin{CON}\label{knkniucbruibckj}
(for supermanifolds) If $U$ is an open subset of $\Cbb^p$, then $\Uc := U\times\Cbb^{0|q}$ will be an open subset of $\Cbb^{p|q}$ with $\Uc_\red = U$. Now let $\Ufr = \{U, V, W, \ldots\}$ be a collection of open subsets of $\Cbb^p$ and $\{\Uc, \Vc,\Wc, \ldots\}$ a corresponding collection of subsets of $\Cbb^{p|q}$. Furthermore, let $\{\Uc_\Vc, \Uc_\Wc, \ldots\}$ be a collection of subsets where $\Uc_\bt\subset \Uc; \Vc_\bt\subset \Vc$ and so on. The transition data $\rho = \{\rho_{\Uc\Vc}\}$ then consists of a collection of isomorphisms $\rho_{\Uc\Vc} : \Uc_\Vc \stackrel{\cong}{\ra}\Vc_\Uc$ satisfying the cocycle condition, represented here by commutativity of the following diagram:
\begin{align}
\xymatrix{
\ar[dr]_{\rho_{\Uc\Wc}} \Uc_\Vc \cap \Uc_\Wc \ar[rr]^{\rho_{\Uc\Vc}} & & \Vc_\Uc\cap \Vc_\Wc\ar[dl]^{\rho_{\Vc\Wc}}\\
& \Wc_\Uc\cap\Wc_\Vc &
}
\label{mciornoihroi}
\end{align} 
The isomorphisms $\{\rho_{\Uc\Vc}\}$ are required to preserve the $\Zbb_2$-grading, in accordance with Definition $\ref{ncnrnv322rnk}$. If $(x^\mu, \q_a)$ (resp. $(y^\mu, \eta_a)$) denote coordinates on $\Uc$ (resp. $\Vc$), then on the intersection $\Uc\cap\Vc$ we can write:
\begin{align}
y^\mu &= \rho_{\Uc\Vc}^\mu(x, \q) = f_{UV}^\mu(x) + \sum_{|I|>0} f_{UV}^{\mu|2I}\q_{2I}
\label{uicuibviuv78v3}
\\
\eta_a &= \rho_{\Uc\Vc, a}(x, \q) = \zeta_{UV, a}^b(x)~\q_b + \sum_{|I|>0} \zeta_{UV, a}^{2I+1}(x)~\q_{2I+1}
\label{porviorhg894h89}
\end{align}
where $I$ is a multi-index and $|I|$ its length; for $I = (i_1, \ldots, i_n)$ that $\q_I = \q_{i_1}\wedge\cdots\wedge\q_{i_n}$; by $2I$ (resp. $2I+1$) it is meant the multi-indices of even (resp. odd) length. The coefficient functions $\{(f_{UV}^{\mu|2I})\}$ (resp. $\{ (\zeta_{UV, a}^{2I+1})\}$) are holomorphic and defined on the intersection $\Uc\cap \Vc$. 
Then just as for manifolds, a supermanifold $\Xfr_{(M, E)}$ is constructed by setting:\footnote{note, this is well defined up to common refinement of the open covering. In this way $\Xfr_{(M, E)}$ will not depend on this choice of open covering.}
\[
\Xfr_{(M, E)} = \frac{\bigsqcup_{\{\Uc,\Vc,\ldots\}} \Uc}{\{\Uc_\Vc \sim_{\rho_{\Uc\Vc}}\Vc_\Uc\}}.
\]
The transition data $\rho = \{\rho_{\Uc\Vc}\}$ is graded by the degree of the monomials $\q_I$. In this way, observe that the transition data for the manifold $M$ (resp. vector bundle $E$) is encoded in the degree-zero (resp. degree-one) components of $\rho$. For more details on this construction see \cite[p. 92]{ROG}. 
\end{CON}

\begin{DEF}\label{k49fk90fk3p}
\emph{Gluing data for a supermanifold $\Xfr$ described in Construction \ref{knkniucbruibckj} will be referred to as a \emph{trivialisation} and denoted $(\Ufr, \rho)$, or simply $\rho$. 
}
\end{DEF}

Following Construction \ref{nruiciurhf94h9h} and \ref{knkniucbruibckj} above, we describe gluing data for thickenings in what follows. Central to the construction is the dichotomy between trivial and non-trivial thickenings from Definition \ref{deduieuidh4hd487h}.

\begin{CON}\label{opnciornvruo4}
Fix subsets $\{U, V, W, \ldots\}$ of $\Cbb^p$. For the open set $U$, let $\Cc_U$ denote its structure sheaf. Then associated to each $U$ is the locally ringed space $\Uc^{(m)} = (U, \Oc_U^{(m)})$, where $\Oc_U^{(m)} = \Cc_U\otimes\wedge^{k\leq m}\Cbb^q$.\footnote{Here we define: $\wedge^{k\leq l}\Cbb^q :=\wedge^\bt\Cbb^q /J^{l+1}$ where $J  = \ker\{\wedge^\bt\Cbb^q \twoheadrightarrow \Cbb\}$. As $\Cbb$-modules (complex vector spaces) we have $\wedge^{k\leq l}\Cbb^q\cong \bigoplus_{k=0}^l\wedge^k\Cbb^q$.
} By Definition \emph{\ref{deduieuidh4hd487h}} we see that $\Uc^{(m)}$ will be a trivial, $m$-th order thickening of $U$. In this way, we see that associated to $\{U, V, W, \ldots\}$ will be a collection of trivial, $m$-th order thickenings $\Ufr^{(m)} = \{\Uc^{(m)}, \Vc^{(m)},\Wc^{(m)}, \ldots\}$. Now just as in Construction \emph{\ref{knkniucbruibckj}} we fix a collection of subsets $\{\Uc_\Vc^{(m)}, \Uc_\Wc^{(m)}, \ldots\}$, where $\Uc_\bt^{(m)}\subset \Uc^{(m)}$, and a collection of isomorphisms $\rho^{(m)} = \{\rho^{(m)}_{\Uc\Vc}\}$, where $\rho^{(m)}_{\Uc\Vc} : \Uc_\Vc^{(m)} \stackrel{\cong}{\ra}\Vc_\Uc^{(m)}$ satisfies the cocycle condition as in \eqref{mciornoihroi}$:$
\begin{align}
\xymatrix{
\ar[dr]_{\rho^{(m)}_{\Uc\Wc}} \Uc^{(m)}_\Vc \cap \Uc^{(m)}_\Wc \ar[rr]^{\rho^{(m)}_{\Uc\Vc}} & & \Vc^{(m)}_\Uc\cap \Vc^{(m)}_\Wc\ar[dl]^{\rho^{(m)}_{\Vc\Wc}}\\
& \Wc^{(m)}_\Uc\cap\Wc^{(m)}_\Vc &
}
\label{090038083j}
\end{align}
Just as in \eqref{uicuibviuv78v3} and \eqref{porviorhg894h89} we have here on $\Uc\cap\Vc$
\begin{align*}
y^\mu &= \rho_{\Uc\Vc}^\mu(x, \q) = f_{UV}^\mu(x) + \sum_{0<|I|\leq m} f_{UV}^{\mu|2I}\q_{2I}
\\
\eta_a &= \rho_{\Uc\Vc, a}(x, \q) = \zeta_{UV, a}^b(x)~\q_b + \sum_{0<|I|\leq m} \zeta_{UV, a}^{2I+1}(x)~\q_{2I+1}.
\end{align*}
With this data a super-geometric thickening $\Xfr_{(M, E)}^{(l)}$ is given by:
\[
\Xfr^{(m)}_{(M, E)} = \frac{\bigsqcup_{\Uc^{(m)}\in \Ufr^{(m)}} \Uc^{(m)}}{\{\Uc^{(m)}_\Vc \sim_{\rho^{(m)}_{\Uc\Vc}}\Vc^{(m)}_\Uc\}}.
\]
Note, this is consistent with the construction of thickenings associated to supermanifolds described in Construction \emph{\ref{nruiciurhf94h9h}}.  Of course, the thickening described here does not presume the existence of some supermanifold containing it and so is abstract.
\end{CON}

\begin{DEF}
\emph{Gluing data for an $m$-th order thickening $\Xfr_{(M, E)}^{(m)}$ as described in Construction \ref{opnciornvruo4} will be referred to as a \emph{trivialisation} and denoted $(\Ufr^{(m)}, \rho^{(m)})$. 
}
\end{DEF}

From Construction \ref{opnciornvruo4} it is evident that Question \ref{SUSYthickeningsquestion} may be reformulated at the level of trivialisations as follows: \emph{given a trivialisation $(\Ufr^{(m)}, \rho^{(m)})$ for $\Xfr^{(m)}$, does there exist a trivialisation $(\Ufr^{(m+1)}, \rho^{(m+1)})$ for a thickening $\Xfr^{(m+1)}$ such that each $\Uc^{(m+1)}$ is a thickening of $\Uc^{(m)}$ and $\rho^{(m)} \equiv \rho^{(m+1)}\mod \Jc^{m+1}$?} We address this question in what follows for $m = 2$ and will thereby arrive at a proof of Theorem \ref{jnvjknknvkjnvk} in this particular case.

\subsection{Proof of Theorem $\ref{jnvjknknvkjnvk}_{m=2}$}
We begin with the following result whose proof can be found, albeit under a different guise, in \cite{GREEN, YMAN}.

\begin{LEM}\label{pjiroich4h894ho}
There exists a bijection $\mathrm{\bf T}^1(\Pi E^{(1)})\cong H^1(M, \Tfr_M\otimes\wedge^2\Ec)$. 
\end{LEM}

\begin{proof}
Let $\Xfr^{(2)}_{(M, E)}$ be a second order thickening with trivialisation $(\Ufr^{(2)}, \rho^{(2)})$. Here $\rho^{(2)} = \{\rho^{(2)}_{\Uc\Vc}\}$ and, following Construction \ref{opnciornvruo4}, we write:
\begin{align}
\rho_{\Uc\Vc}^{\mu;(2)} = f_{UV}^\mu + f_{UV}^{\mu|ij}~\q_{ij}
&&
\mbox{and}
&&
\rho_{\Uc\Vc, a}^{(2)} = \zeta_{UV, a}^b~\q_b.
\label{piorf84hf4j0f43}
\end{align}
The summation in \eqref{piorf84hf4j0f43} over the indices $i, j$ is implicit. Now as a result of the cocycle condition we will find that $f^{(2)} = \{(f^{\mu|ij}_{UV})\}$ is a cocycle representative of a class $\om_{(M, E)}\in H^1(M, \wedge^2\Ec\otimes\Tfr_M)$. Furthermore, it is a straightforward check that the notion of equivalence in Definition \ref{equivextensusythick} is precisely captured by that notion of equivalence at the level of cohomology. That is, if $\Xfr^{(2)\p}_{(M, E)}$ is another first order thickening of $\Pi E^{(1)}$ and if the corresponding cocycle representative $f^{(2)\p}$ is cohomologous to $f^{(2)}$, then it is straightforward to construct an equivalence between the thickenings $\Xfr^{(2)}_{(M, E)}$ and $\Xfr^{(2)\p}_{(M, E)}$ and vice-versa. For the calculations justifying our assertions, we refer to \cite[Chapter 2]{BETTPHD}. 
\end{proof}

Now if we are given a second order thickening $\Xfr^{(2)}_{(M, E)}$, then by construction it will be a thickening of $\Pi E^{(1)}$ by Lemma \ref{fiugf78g79h380fj09}. A such, by Lemma \ref{pjiroich4h894ho}, it will define a class in $H^1(M, \Tfr_M\otimes\wedge^2\Ec)$. We denote this class by $\om_{(M, E)}$ and refer to it as the \emph{obstruction class associated to $\Xfr^{(2)}_{(M, E)}$}. The central result of the present section, which we will prove en route to Theorem \ref{jnvjknknvkjnvk}, is the following:

\begin{PROP}\label{popiioh4o8fh84o}
Consider the exact sequence associated to the degree-two component of tangent sheaf $\Tfr_{\Pi E}$ in Lemma \emph{\ref{hjbbjhhveyjj}}, i.e., the sequence in \eqref{jkkbkbkbk} at $k = 2$. It induces a long-exact sequence on sheaf cohomology containing the piece:
\begin{align}
\xymatrix{
\cdots \ar[r] & H^1(M, \Tfr_M\otimes\wedge^2\Ec) \ar[r]^{\pt_*} & H^2(M, \wedge^3\Ec\otimes\Ec^\vee) \ar[r] & \cdots
}
\label{poijirjciojj49j4}
\end{align}
Then a given, second-order thickening $\Xfr^{(2)}_{(M, E)}$ is unobstructed if and only if its obstruction class satisfies: $\pt_*(\om_{(M, E)})=0$. 
\end{PROP}

We will now give a proof of Theorem \ref{jnvjknknvkjnvk} for $m = 2$ by appealing to the description of a thickening by means of trivialisations, as in Construction \ref{opnciornvruo4}. 
\\\\
\emph{Proof of Theorem $\emph{\ref{jnvjknknvkjnvk}}_{m=2}$}. Following on from the proof of Lemma \ref{pjiroich4h894ho}, we consider here the analogous picture for any thickening $\Xfr^{(3)}_{(M, E)}$ of $\Xfr^{(2)}_{(M, E)}$. If $(\Ufr^{(3)}, \rho^{(3)})$ is a trivialisation for $\Xfr^{(3)}_{(M, E)}$, then $\rho^{(3)}$ will be as in \eqref{piorf84hf4j0f43}, however we will have an additional term:
\begin{align}
\rho^{(3);\mu}_{\Uc\Vc} = \rho^{(2);\mu}_{\Uc\Vc}
&&
\mbox{and}
&&
\rho^{(3)}_{\Uc\Vc, a} = \rho^{(2)}_{\Uc\Vc, a} + \zeta_{UV, a}^{ijk}~\q_{ijk}.
\label{piorf83f33f34hf4j0f43}
\end{align}
The summation in \eqref{piorf83f33f34hf4j0f43} over the indices $i, j, k$ is implicit. In assuming that a thickening $\Xfr^{(3)}_{(M, E)}$ of $\Xfr^{(2)}_{(M, E)}$ exists, we will be faced with a non-trivial condition that needs to hold. We will write down this condition explicitly here as an equation which, we argue, will lie in the second cohomology group. This is the assertion of the present theorem. To begin note that $\zeta^{(3)} = \{(\zeta_{UV, a}^{ijk})\}$ here is a $1$-cochain valued in $\wedge^3\Ec\otimes\Ec^\vee$. In imposing the cocycle condition on $\rho^{(3)}$ we obtain:  
\begin{align}
\zeta_{UW, a}^{lmn}~\q_{lmn}\frac{\pt}{\pt \xi_a} - \zeta_{UV, b}^{lmn}~\q_{lmn}\frac{\pt}{\pt \eta_b} 
~-~
&\zeta_{VW, a}^{ijk}~\eta_{ijk}\frac{\pt}{\pt \xi_a}
\label{ioiohfojfoifjof}
\\
&=
\notag
\\
\frac{\pt \zeta_{VW, a}^b}{\pt y^\mu}&f_{UV}^{\mu|lm}\zeta_{UV, b}^n~\q_{lmn}\frac{\pt}{\pt \xi_a}.
\label{ejckencjkencjkenc}
\end{align}
Now note that $\eqref{ejckencjkencjkenc}$ is a 2-cochain. The proof of Theorem $\ref{opnciornvruo4}_{m=2}$ will  follow if we can successfully argue that $\eqref{ejckencjkencjkenc}$ is in fact a 2-cocycle, for then the equality $\eqref{ioiohfojfoifjof} = \eqref{ejckencjkencjkenc}$ will make sense as a statement in the second cohomology group, which is precisely what is asserted in Theorem $\ref{jnvjknknvkjnvk}_{m=2}$. We will now show that \eqref{ejckencjkencjkenc} is a 2-cocycle in a way that will also prove Proposition \ref{popiioh4o8fh84o} with the following:

\begin{LEM}\label{lkliioolkss}
$\eqref{ejckencjkencjkenc}$ is a $2$-cocycle representative of $\pt_*(\om_{(M, E)})$. 
\end{LEM}

\begin{proof}
The proof follows essentially by construction. Indeed, firstly note that the map $\Tfr_{\Pi E}[2]\ra \Tfr_M\otimes\wedge^2\Ec$ is given by sending a vector field $v$ in $\Tfr_{\Pi E}[2]$ to $v\mod \Jc^3$. At the level of $1$-cocycles: let $f^{(2)} =\{f^{(2)}_{\Uc\Vc}\} =  \{(f_{\Uc\Vc}^{\mu|ij})\}$ be a cocycle representative for $\om_{(M, E)}$. Then we have the cocycle condition:
 \begin{align}
f_{UW}^{\s|kl}~\q_{kl}\frac{\pt}{\pt z^\s} =   f^{\tau|kl}_{UV}~\q_{kl} \frac{\pt}{\pt y^\tau}+ f^{\s|ij}_{VW}~\eta_{ij}\frac{\pt}{\pt z^\s}
\label{4h9f8h5hiurhh89}
\end{align}
Now let $\{Y_{\Uc\Vc}\}$ be such that $Y_{\Uc\Vc}\equiv f^{(2)}_{\Uc\Vc}$ modulo $\Jc^3$. Then we have the 2-cocycle $Y_{\Uc\Vc\Wc} = \{Y_{\Uc\Wc} - Y_{\Uc\Vc} - Y_{\Vc\Wc}\}$ which, by \eqref{4h9f8h5hiurhh89}, vanishes modulo $\Jc^3$. Now by exactness of the induced, long-exact sequence on cohomology in \eqref{poijirjciojj49j4} we have, at the level of 2-cocycles, a 2-cocycle $Z = \{Z_{\Uc\Vc\Wc}\}$ in $\mathcal Z^2(\Ufr, \wedge^3\Ec\otimes\Ec^\vee)$ mapping to $Y = \{Y_{\Uc\Vc\Wc}\}$. That is $Z = Y|_{\wedge^3\Ec\otimes\Ec^\vee}$. Its cohomology class $[Z]$ is defined to be $\pt_*(\om_{(M, E)})$. To get an expression for its representative we will make use of the transition data $\rho^{(2)}$. Firstly recall that $f^{(2)}_{\Uc\Vc} = (f^{\mu|ij}_{\Uc\Vc})$ defines a vector field on $\Vc_\Uc$. With $\rho^{(2)}_{\Vc\Wc} : \Vc^{(2)}_\Wc \stackrel{\cong}{\ra}\Wc^{(2)}_\Uc$, we can send $f^{(2)}_{\Uc\Vc}$ to a vector field $(\rho^{(2)}_{\Vc\Wc})_*f^{(2)}_{\Uc\Vc}$ on $\Wc_\Uc$. Now note that $\rho^{(2)} = \rho^{(3)}\mod \Jc^3$. In writing $\rho^{(2)}$ in terms of its even and odd components $\rho^{(2)} = (\rho^{(2);+}, \rho^{(2);-})$, the condition in \eqref{4h9f8h5hiurhh89} asserts, on $\Wc_\Uc$, that: $f_{\Uc\Wc}^{(2)} = (\rho_{\Uc\Vc}^{(2);+})_*f^{(2)}_{\Uc\Vc} + f^{(2)}_{\Vc\Wc}$. Hence, as a vector field on $\Wc_\Uc$, we have:
\begin{align}
 Y_{\Uc\Vc\Wc}
&= (\rho^{(2);-}_{\Vc\Wc})_*f^{(2)}_{\Uc\Vc} 
\label{9g859gh95hfFD}
\\
\notag
&=f^{\mu|ij}_{UV}\q_{ij} ~(\rho^{(2)}_-)_*\left( \frac{\pt}{\pt y^\mu}\right)
\\
\notag
&= f^{\mu|ij}_{UV}\q_{ij} \cdot \left( \frac{\pt \zeta_{VW, a}^b}{\pt y^\mu}\eta_b\frac{\pt}{\pt \xi_a}\right)
\\
\notag
&=f^{\mu|ij}_{UV}\frac{\pt \zeta_{VW, a}^b}{\pt y^\mu}\zeta_{UV, b}^c~ \q_{ijc} \frac{\pt}{\pt \xi_a}.
\end{align} 
The lemma now follows by comparing \eqref{ejckencjkencjkenc} with \eqref{9g859gh95hfFD}. 
\end{proof}

An alternative and more direct check that \eqref{ejckencjkencjkenc} satisfies the conditions to be a 2-cocycle can be found in \cite[Appendix A.1]{BETTPHD}. The proof of Theorem $\ref{jnvjknknvkjnvk}_{m=2}$ and Proposition \ref{popiioh4o8fh84o} now follows. \qed
\\

A much simpler proof of Theorem $\ref{jnvjknknvkjnvk}_{m=2}$ may be given more indirectly and independently of the direct verification in Lemma \ref{lkliioolkss}. This uses the assumption that $\Xfr^{(2)}_{(M, E)}$ is a second order thickening, meaning $\rho^{(2)}$ must satisfy the cocycle condition. The proof is inductive and can in fact be applied to give a full proof of Theorem \ref{jnvjknknvkjnvk}, which we do so in what follows. 

\subsection{An Elementary Proof of Theorem $\ref{jnvjknknvkjnvk}$}
We will consider the case where we are given a super-geometric thickening of odd order $\Xfr^{(2m-1)}$. The argument remains essentially unchanged in the even case so we omit it here. Firstly, to briefly recap: if $\Xfr^{(2m)}$ denotes a $(2m)$-th order, super-geometric thickening containing $\Xfr^{(2m-1)}$, then it will admit a trivialisation $(\Ufr^{(2m)}, \rho^{(2m)})$, with $\rho^{(2m)} = \{\rho^{(2m)}_{\Uc\Vc}\}$ subject to the cocycle condition in \eqref{090038083j}, i.e., that:
\begin{align}
\rho^{(2m);\mu}_{\Uc\Wc} = \rho^{(2m);\mu}_{\Vc\Wc}\circ \rho^{(2m)}_{\Uc\Vc}.
\label{dkncrcurcrknckmcr87h3}
\end{align}
Note that $\rho_{\Uc\Vc, a}^{(2m)} = \rho_{\Uc\Vc, a}^{(2m-1)}$ and that, by assumption, $\rho^{(2m-1)}$ will satisfy the cocycle condition. Now write,
\begin{align}
\rho_{\Uc\Vc}^{(2m);\mu} = \rho_{\Uc\Vc}^{(2m-2);\mu} + \phi_{\Uc\Vc}^\mu
\label{pprirjiovrio}
\end{align}
where $\phi = \{(\phi_{\Uc\Vc}^\mu)\}$ comprises homogeneous quantities of degree-$(2m)$. In applying the decomposition in \eqref{pprirjiovrio} to \eqref{dkncrcurcrknckmcr87h3} we arrive at the following expression:
\begin{align}
\rho_{\Uc\Wc}^{(2m-2);\mu} - \rho_{\Vc\Wc}^{(2m-2);\mu}\circ \rho_{\Uc\Vc}^{(2m-2)}
=
\phi^\mu_{\Vc\Wc}\circ \rho_{\Uc\Vc}
+
\frac{\pt f^\mu_{VW}}{\pt y^\nu}\phi_{\Uc\Vc}^\nu
-
\phi^\mu_{\Uc\Wc}.
\label{dcnn8984jid}
\end{align}
We see that the right-hand side of \eqref{dcnn8984jid} consists only of those quantities which regulate the extension of $\rho^{(2m-1)}$ to $\rho^{(2m)}$. As for the left-hand side of \eqref{dcnn8984jid}, we set:
\begin{align}
\Gam_{\Uc\Vc\Wc}^\mu := \rho_{\Uc\Wc}^{(2m-2);\mu} - \rho_{\Vc\Wc}^{(2m-2);\mu}\circ \rho_{\Uc\Vc}^{(2m-2)}.
\label{mchrehehruei38deu383j}
\end{align}
Since $\rho^{(2m-2)}$ satisfies the cocycle condition (modulo $\Jc^{2m}$), it follows that $\Gam_{\Uc\Vc\Wc}^\mu$ will be homogeneous and of degree-$(2m)$.\footnote{Recall that we are implicitly considering all quantities here modulo $\Jc^{2m+1}$.} 

\begin{LEM}\label{bipalabophhhh89}
Let 
\[
\Gam_{\Uc\Vc\Wc} = \Gam_{\Uc\Vc\Wc}^\mu\frac{\pt}{\pt z^\mu}.
\] 
Then $\Gam = \{\Gam_{\Uc\Vc\Wc}\}$ will be a $2$-cocycle valued in $\Tfr_M\otimes\wedge^{2m}\Ec$.
\end{LEM}

\begin{proof}
The proof follows along essentially the same lines as that given in \cite[p. 254-255]{KS}, so we omit it here.
\end{proof}

Finally, from \eqref{dcnn8984jid} we obtain the following equation:
\begin{align}
\phi^\mu_{\Vc\Wc}
\frac{\pt}{\pt z^\mu}
+
\phi_{\Uc\Vc}^\nu
\frac{\pt}{\pt y^\nu}
-
\phi^\mu_{\Uc\Wc}
\frac{\pt}{\pt z^\mu}
= 
\Gam_{\Uc\Vc\Wc}^\mu\frac{\pt}{\pt z^\mu}.
\label{dmodj833jso}
\end{align}
The left-hand side of \eqref{dmodj833jso} is clearly the coboundary of a 1-cochain valued in $\Tfr_M\otimes\wedge^{2m}\Ec$, whereas the right-hand side is a 2-cocycle valued in $\Tfr_M\otimes\wedge^{2m}\Ec$ by Lemma \ref{bipalabophhhh89}. The theorem now follows. \qed

\begin{REM}
\emph{Our usage of the word `elementary' to term the proof of Theorem \ref{jnvjknknvkjnvk} given here is motivated by a similar usage of this word in \cite{KS}. There it is used in reference to a proof of the existence of (formal) deformations of complex structures on compact manifolds. 
}
\end{REM}

\begin{REM}
\emph{Up to sign, the equation in \eqref{dmodj833jso} is explicitly written down as the equality $\eqref{ioiohfojfoifjof} = \eqref{ejckencjkencjkenc}$ in the case where we are given an even thickening $\Xfr^{(2)}_{(M, E)}$.}
\end{REM}

\subsection{A Classification of Thickenings}
Consequences of Theorem \ref{jnvjknknvkjnvk} and its proof provided in the previous (sub)section are the results to be given here regarding classifications of thickenings. It will be convenient to introduce the following notation:
\begin{align}
\mathcal Q^{(k);+}_E := 
\Tfr_M\otimes \wedge^k\Ec
&&
\mbox{and}
&&
\mathcal Q^{(k);-}_E := \wedge^k\Ec\otimes\Ec^\vee.
\label{dafoprmijo4jo4}
\end{align}
For further convenience we will write:
\[
\mathcal Q_E^{(i);\pm} := 
\left\{\begin{array}{ll}
\mathcal Q_E^{(i);+} & \mbox{if $i$ is even}
\\
\mathcal Q_E^{(i);-} &\mbox{if $i$ is odd}.
\end{array}
\right.
\]
With this notation the exact sequence in Lemma \ref{hjbbjhhveyjj} becomes:
\begin{align}
0 \ra \mathcal Q_E^{(k+1);-} \hookrightarrow  \Tfr_{\Pi E}[k] \twoheadrightarrow \mathcal Q_E^{(k);+}\ra0.
\label{oinriobuoioie}
\end{align}
In particular, we see that $\mathcal Q_E^{(i);\pm}$ will always be affiliated with sections of the tangent sheaf of even-degree tangent vectors.
\\\\
We have so far been guided by Question \ref{SUSYthickeningsquestion}, which was addressed in part in  Theorem \ref{jnvjknknvkjnvk}. In what follows we present a result which can be seen as a generalisation of Proposition \ref{popiioh4o8fh84o} or as a slight furthering of Theorem \ref{jnvjknknvkjnvk}. In particular, it will address Question \ref{SUSYthickeningsquestion} more directly. 

\begin{PROP}\label{o3jdioj3ijp3}
Let $\Xfr^{(l)}_{(M, E)}$ be a thickening of order $l$ and suppose $\mathrm{\bf T}^1(\Xfr_{(M, E)}^{(l)})$ is non-empty. Then there exists a map 
\[
\pt_* : \mathrm{\bf T}^1(\Xfr^{(l)}_{(M, E)}) \lra H^2(M, \mathcal Q^{(l+2);\pm}_E)
\]
which measures the failure for a thickening of $\Xfr^{(l)}_{(M, E)}$ to be unobstructed.   
\end{PROP}

\begin{proof}
The map $\pt_*$ can be described explicitly by appealing to a trivialisation. To do so, let $\Xfr^{(l)}_{(M, E)}$ be a thickening of order $l$ and $\Xfr^{(l+1)}_{(M, E)}$ a thickening of $\Xfr^{(l)}_{(M, E)}$. We think of $\Xfr_{(M,E)}^{(l+1)}$ as a representative of a class $[\Xfr_{(M, E)}^{(l+1)}]\in \mathrm{\bf T}^1(\Xfr_{(M, E)}^{(l)})$. Let $(\Ufr^{(l+1)}, \rho^{(l+1)})$ be a trivialisation for $\Xfr_{(M, E)}^{(l+1)}$. Then associated to $\rho^{(l+1)}$ will, by Lemma \ref{bipalabophhhh89}, be the 2-cocycle $\{\Gam^{(l+2)}_{\Uc\Vc\Wc}\}\in \mathcal Z^2(\Ufr, \mathcal Q^{(l+2);\pm}_E)$. We set:
\begin{align}
\pt_*\left( [\Xfr^{(l+1)}_{(M, E)}]\right) := \left[ \left\{\Gam^{(l+2)}_{\Uc\Vc\Wc}\right\}\right].
\label{mdio3nfuih3fh3o}
\end{align}
The goal now is to argue that \eqref{mdio3nfuih3fh3o} is well-defined. To that end, let $\tilde\Xfr^{(l+1)}_{(M, E)}$ be another thickening of $\Xfr^{(l)}_{(M, E)}$ and suppose $\Xfr^{(l+1)}_{(M, E)}\sim \tilde\Xfr^{(l+1)}_{(M, E)}$ as thickenings. This means there exists an isomorphism $\lam: \Xfr^{(l+1)}_{(M, E)}\cong \tilde\Xfr^{(l+1)}_{(M, E)}$ which is the identity when restricted to $\Xfr^{(l)}_{(M, E)}$. Let $(\Ufr^{(l+1)}, \tilde\rho^{(l+1)})$ be a trivialisation for $\tilde\Xfr^{(l+1)}_{(M, E)}$ and denote by $\{\tilde\Gam^{(l+2)}_{\Uc\Vc\Wc}\}$ the 2-cocycle constructed from the trivialisation $\tilde\rho^{(l+1)}$. In order to show that the map $\pt_*$ in \eqref{mdio3nfuih3fh3o} is well-defined, we need to show that $[\{\Gam^{(l+2)}_{\Uc\Vc\Wc}\}]= [\{\tilde\Gam^{(l+2)}_{\Uc\Vc\Wc}\}]$ in $H^2(M, \mathcal Q^{(l+2);\pm}_E)$. Under the assumptions here we will show that the cocycle representatives coincide, i.e., that $\Gam^{(l+2)}_{\Uc\Vc\Wc} = \tilde\Gam^{(l+2)}_{\Uc\Vc\Wc}$. So firstly, as we assume $\Xfr^{(l+1)}_{(M, E)}\sim \tilde\Xfr^{(l+1)}_{(M, E)}$, we have on the intersection $\Uc\cap\Vc$ that
\begin{align}
\rho_{\Uc\Vc}^{(l+1)}\circ \lam_\Uc= \lam_\Vc\circ \tilde\rho^{(l+1)}_{\Uc\Vc}
\label{4903jf903jf093}
\end{align}
where $\{\lam_\Uc\}$ is a 0-cochain representative of $\lam: \Xfr^{(l+1)}_{(M, E)}\cong \tilde\Xfr^{(l+1)}_{(M, E)}$ with respect to the cover $\Ufr^{(l+1)}$. Now since $\lam_\Uc : \Uc^{(l+1)}\cong \Uc^{(l+1)}$ is trivial modulo $\Jc^l$ (i.e., that $\lam_\Uc|_{\Uc^{(l)}}$ is the identity), we can construct from $\{\lam_\Uc\}$ a $\mathcal Q_E^{(l+1);\pm}$-valued, 0-cochain $\{\nu_\Uc\}$.\footnote{a simpler result justifying our assertions here is a certain exact sequence of sheaves of groups which we discuss in the section to come (Proposition \ref{dm30j8fjf3o}).} Expanding \eqref{4903jf903jf093} will now reveal that:
\begin{align}
\rho^{(l+1)}_{\Uc\Vc} - \tilde\rho^{(l+1)}_{\Uc\Vc} = (\dt\nu)_{\Uc\Vc}
=
\nu_\Vc-\nu_\Uc
\label{fj4j49jf94kf4kp}
\end{align}
and as a result:
\begin{align}
\tilde \Gam^{(l+2)}_{\Uc\Vc\Wc} &= \tilde \rho^{(l+1)}_{\Uc\Wc} - \tilde\rho^{(l+1)}_{\Wc\Vc}\circ \tilde\rho^{(l+1)}_{\Uc\Vc}
\notag
\\
&= \lam_\Wc^{-1}\circ \rho^{(l+1)}_{\Uc\Wc}\circ \lam_\Uc - \left(\lam_\Wc^{-1}\circ\rho^{(l+1)}_{\Vc\Wc}\circ\lam_\Vc\right)\circ \left(\lam_\Vc^{-1}\circ  \rho^{(l+1)}_{\Uc\Vc}\circ \lam_\Uc\right)
\notag
\\
&= 
\Gam^{(l+2)}_{\Uc\Vc\Wc} - (\dt\nu)_{\Uc\Wc} + (\dt\nu)_{\Vc\Wc} + (\dt\nu)_{\Uc\Vc}
\notag
\\
&= \Gam^{(l+2)}_{\Uc\Vc\Wc}.
\label{39j03jf30jfp3}
\end{align}
Hence we see that if $\Xfr_{(M, E)}^{(l+1)}\sim \tilde \Xfr_{(M, E)}^{(l+1)}$, then the 2-cocycles $\Gam_{\Uc\Vc\Wc}$ and $\tilde \Gam_{\Uc\Vc\Wc}$ associated to trivialisations $\rho^{(l+1)}$ and $\tilde\rho^{(l+1)}$ of $\Xfr_{(M, E)}^{(l+1)}$ and $\tilde \Xfr_{(M, E)}^{(l+1)}$ respectively will coincide. In particular, the map $\pt_*$ in \eqref{mdio3nfuih3fh3o} is well-defined. Now, from the proof of Theorem \ref{jnvjknknvkjnvk} given earlier, we see that $\Xfr^{(l+1)}_{(M, E)}$ will be an obstructed thickening if and only if $\pt_*( [\Xfr^{(l+1)}_{(M, E)}])\neq0$. It is in this sense that $\pt_*$ measures the failure for a given first-order extension of $\Xfr^{(l)}$ to be unobstructed. The proposition now follows.
\end{proof}

Recall that $\Xfr_{(M, E)}$ will admit a trivialisation by Construction \ref{opnciornvruo4} which is unique up to common refinement. Should we consider trivialisations of $\Xfr_{(M, E)}$ directly, then from the proof of Proposition \ref{o3jdioj3ijp3} above---and in particular from \eqref{39j03jf30jfp3}, we obtain a stronger result:

\begin{COR}\label{eiooihf839hf3ho}
Let $\Xfr_{(M, E)}^{(l)}$ be an unobstructed thickening and suppose $(\Ufr^{(l)}, \rho^{(l)})$ is a trivialisation for $\Xfr_{(M, E)}^{(l)}$. Denote by $\mathrm{\bf T}^1(\Ufr^{(l)}, \rho^{(l)})$ the set of thickenings of $\Xfr_{(M, E)}^{(l)}$ equipped with this choice of trivialisation $(\Ufr^{(l)}, \rho^{(l)})$ and suppose it is non-empty. Then the following diagram commutes:
\begin{align}
\xymatrix{
\ar[dr]_{\pt_*}\mathrm{\bf T}^1(\Ufr^{(l)}, \rho^{(l)}) \ar[r] & \mathcal Z^2(\Ufr, \mathcal Q_E^{(l+2);\pm})\ar[d]
\\
& \mbox{\emph{ \v H}}^2(\Ufr, Q_E^{(l+2);\pm})
}
\label{prfpejpj3pj390}
\end{align}
where in \eqref{prfpejpj3pj390}: the horizontal map is given by sending a trivialisation $(\Ufr^{(l+1)}, \rho^{(l+1)})$ to its obstruction class $\Gam^{(l+2)}$; the vertical arrow is the natural map on \v Cech cohomology of the covering $\Ufr$; and $\pt_*$ is the map from Proposition $\ref{o3jdioj3ijp3}$.\qed
\end{COR}

We turn our attention now to the set of thickenings itself with the intent to show that, in general, it will admit the structure of a pseudo-torsor over a certain group. The notion of a torsor itself comes up in many areas of mathematics and the definition we present below is taken from \cite[p. 49]{HARTDEF}.

\begin{DEF}
\emph{A set $S$ is said to be a \emph{torsor} for a group $G$ if there exists an action of $G$ on $S$, i.e., a set-theoretic map $G\times S\ra S$, such that:
\begin{enumerate}[(i)]
	\item for each $g\in G$, the map $g\cdot : S\ra S$ given by $s\mapsto g\cdot s$, $s\in S$, is bijective and;
	\item $S$ is non-empty.
\end{enumerate}
If $S$ is empty, then it is said to be a \emph{pseudo-torsor}.}
\end{DEF}

Following on from Proposition \ref{o3jdioj3ijp3} we have:

\begin{THM}\label{o4993uf0oijwoijfwojpw}
Given an $l$-th order thickening $\Xfr_{(M, E)}^{(l)}$, then the set of thickenings $\mathrm{\bf T}^1(\Xfr_{(M, E)}^{(l)})$ is a pseudo-torsor for the group $H^1(M, \mathcal Q_E^{(l+1);\pm})$.
\end{THM}

\begin{proof}
Assume $\mathrm{\bf T}^1(\Xfr_{(M, E)}^{(l)})$ is non-empty. This means $\Xfr_{(M, E)}^{(l)}$ is unobstructed. Then associated to any thickening $\Xfr_{(M, E)}^{(l+1)}$ will be a trivialisation $(\Ufr^{(l+1)}, \rho^{(l+1)})$ and associated to $\rho^{(l+1)}$ will be 1-cochains $c^{(j)}(\rho)=\{c_{\Uc\Vc}^{(j)}\}\in C^1(\Ufr^{(l+1)}, \mathcal Q_E^{(j);\pm})$ for $j = 2,\cdots l+1$. On the intersection $\Uc\cap\Vc$ the term $c_{\Uc\Vc}^{(j)}$ is precisely the homogeneous, degree-$j$ component of $\rho_{\Uc\Vc}^{(l+1)}$. Now recall that $\Xfr_{(M, E)}^{(l)}$ is unobstructed from our assumption. As such we know from Proposition \ref{o3jdioj3ijp3} that the 2-cocycle $\Gam^{(l+1)}$ corresponding to $\rho^{(l+1)}$ will be cohomologically trivial. Now let $\al^{(l+1)}\in \mathcal Z^1(\Ufr, \mathcal Q_E^{(l+1);\pm})$. Then $\dt\al^{(l+1)} = 0$. Define $\tilde\rho^{(l+1)} = \rho^{(l+1)} + \al^{(l+1)}$. Evidently the 2-cocycle $\tilde\Gam$ associated to $\tilde\rho^{(l+1)}$ is
\[
\tilde\Gam_{\Uc\Vc\Wc} = \Gam_{\Uc\Vc\Wc} +(\dt\al)_{\Uc\Vc\Wc} = \Gam_{\Uc\Vc\Wc}
\]
and since $\Gam_{\Uc\Vc\Wc}\sim 0$, we find $\tilde  \Gam_{\Uc\Vc\Wc}\sim 0$ also. Hence $(\Ufr^{(l+1)}, \tilde\rho^{(l+1)})$ defines a trivialisation from which we can construct another, unobstructed thickening $\tilde \Xfr_{(M, E)}^{(l+1)}$ of $\Xfr^{(l)}$ as in Construction \ref{opnciornvruo4}. From here, it is not too hard to deduce that $\Xfr_{(M, E)}^{(l+1)}$ and $\tilde \Xfr_{(M, E)}^{(l+1)}$ will be equivalent as thickenings of $\Xfr_{(M, E)}^{(l)}$ if and only if $\al^{(l+1)}\sim 0$. If $\Xfr^{(l)}_{(M, E)}$ is obstructed then by definition there will not exist any thickening containing it, which means $\mathrm{\bf T}^1(\Xfr_{(M, E)}^{(l)})=\eset$. The present theorem now follows. For more details see \cite[Chapter 2]{BETTPHD}.
\end{proof}

Then as a generalisation of Lemma \ref{pjiroich4h894ho} we have:

\begin{COR}\label{dopmropvopejpv}
$\mathrm{\bf T}^1(\Pi E^{(l)})\cong H^1(M, \mathcal Q_E^{(l+1);\pm})$.
\qed
\end{COR}

\section{A Decomposition of the Obstruction Spaces}\label{5}

\noindent
It was mentioned in the previous section that the discussion and results so far presented were guided by Question \ref{SUSYthickeningsquestion}. This question was answered in Theorem \ref{jnvjknknvkjnvk} and Proposition \ref{o3jdioj3ijp3}. In the present section we consider the following variant incorporating supermanifolds:

\begin{QUE}\label{kd93f3jf3kf9p4}
Suppose $\Xfr^{(k)}$ is a given, unobstructed thickening. Then will there exist some supermanifold $\Xfr$ containing $\Xfr^{(k)}$?
\end{QUE}

In general the answer to Question \ref{kd93f3jf3kf9p4} above will be in the negative and, just as in Proposition \ref{o3jdioj3ijp3}, we can identify conditions under which an unobstructed thickening will be associated to a supermanifold. Before doing so it will be necessary to briefly deliberate on obstruction theory for supermanifolds.

\subsection{Preliminaries: Obstruction Theory for Supermanifolds}
Let $\Xfr_{(M, E)}$ be a complex supermanifold with split model $\Pi E$ (see Section \ref{fjf9j490fj40fj40}). Then by Definition \ref{ncnrnv322rnk} we know that $\Xfr_{(M, E)}$ and $\Pi E$ will be \emph{locally} isomorphic.

\begin{DEF}
\emph{The supermanifold $\Xfr_{(M, E)}$ is said to be \emph{split} if it is globally isomorphic to its split model.}
\end{DEF}

With this definition of splitness we can ask the following natural question guiding considerations in obstruction theory for supermanifolds:

\begin{QUE}\label{kf93j890jf0j}
Let $\Xfr$ be a supermanifold. Is it split?
\end{QUE}

As a smooth supermanifold $\Xfr_{(M, E)}$ is split, a result which is well known and originally due to Batchelor in \cite{BAT}. In the holomorphic setting, a given complex supermanifold certainly need not be split and the study of Question \ref{kf93j890jf0j} in this setting is referred to as `obstruction theory for supermanifolds'.  In this section we present some well known results, following \cite{BER, GREEN, YMAN, ONISHCLASS}. Our starting point is in the construction of the sheaf of groups $\Gc^{(i)}_E$ for any $i\geq2$ as follows: on a complex manifold $M$ let $E\ra M$ be a rank $q$, holomorphic vector bundle and $\Ec$ its sheaf of holomorphic sections. Define,
\begin{align}
\Gc^{(i)}_E := \left\{ \al\in \mathscr Aut\wedge^\bt\Ec\mid \al(u) - u\in \Jc^i\right\}.
\label{fkfugunefefmelkfe}
\end{align}
Then by construction, for $i = 2$, we have the following short-exact sequence of sheaves of groups on $M$ which is right-split, as observed by Onishchik in \cite{ONISHCLASS}:
\[
1 \ra \Gc^{(2)}_E \hookrightarrow \mathscr Aut\wedge^\bt\Ec \twoheadrightarrow \mathscr Aut~\Ec\ra1.
\] 
For $i>2$ the relation of these sheaves $\Gc^{(i)}_E$ to $\mathcal Q_E^{(i);\pm}$ from \eqref{dafoprmijo4jo4} is captured in the following result by Green in \cite{GREEN}:

\begin{PROP}\label{dm30j8fjf3o}
For each $i$ the sheaf $\Gc^{(i)}_E$ contains $\Gc^{(i+1)}_E$ as a (sheaf of) normal subgroups. Moreover, for $i> 2$, there exists an isomorphism
\[
\frac{\Gc_E^{(i)}}{\Gc_E^{(i+1)}} \cong 
\mathcal Q_E^{(i);\pm}.
\]
If $i> q$, then $\Gc^{(i)}_E = \{1\}$. \qed
\end{PROP}

\begin{REM}
\emph{As a result of the last sentence in the statement of Proposition \ref{dm30j8fjf3o} we have $\Gc^{(q)}_E \cong \mathcal Q^{(q);\pm}_E$. Hence $\Gc^{(q)}_E $ is abelian.}
\end{REM}

By the general theory of non-abelian sheaf cohomology developed in \cite{GROTHNONAB} and reviewed in \cite[p. 158-62]{BRY}, the results in Proposition \ref{dm30j8fjf3o} may be recast into a short-exact sequence of sheaves of groups from which we obtain a long exact sequence on \v Cech cohomology (as pointed sets). A piece of this sequence is given below:
\begin{align}
\cdots \lra H^1(\Gc^{(i+1)}_E) \lra H^1(\Gc^{(i)}_E) \stackrel{\om_*}{\lra} H^1(\mathcal Q_E^{(i);\pm})
\label{8393hf93j3fojf83}
\end{align}
where $H^i(-) = H^i(M, -)$ above.  Hence, in this way, we find the obstruction spaces $H^1(\mathcal Q_E^{(i);\pm})$ appearing here. Following \cite{DW1} we submit:

\begin{DEF}\label{rfijf903jf903kf}\emph{Elements in $H^1(M, \Gc^{(2)}_E)$ will be referred to as \emph{trivialisations}. Elements in $H^1(M, \Gc^{(i)}_E)$, for $i>2$, will be referred to as \emph{lifts of trivialisations} or \emph{level-$(i+1)$ trivialisations}}
\end{DEF}

\begin{DEF}\label{kd93jf03j0j3j3p}\emph{Elements in the image of $\om_*$ are referred to as \emph{obstruction classes to splitting}.
}
\end{DEF}

By construction, the obstruction classes to splitting depend on a choice of trivialisation or lifts thereof. To see the relevance of these constructs to supermanifold theory, firstly consider the following set:
\begin{align}
\Mfr_{(M, E)} = 
\left\{\mbox{supermanifolds modelled on $(M, E)$}\right\}/\cong
\label{rkfj490fj4fp3okp3ofk3p}
\end{align}
We think of $\Mfr_{(M, E)}$ here as a pointed set, with base-point represented by the split model $\Pi E$. We now have from \cite{GREEN}:

\begin{THM}\label{dn8j83jp3k9k922}
Any supermanifold $\Xfr_{(M, E)}$ will define a trivialisation in $H^1(M, \Gc^{(2)}_E)$. Furthermore, there exists a one-to-one correspondence as pointed sets:
\[
\Mfr_{(M, E)} \cong \frac{H^1(M, \Gc^{(2)}_E)}{H^0(M, \mathscr Aut~\Ec)}.
\]\qed
\end{THM}

As a result of Theorem \ref{dn8j83jp3k9k922} above we deduce the following statements pertaining to Question \ref{kf93j890jf0j}:
\begin{enumerate}[(i)]
	\item any supermanifold $\Xfr_{(M, E)}$ will define an element in $H^1(M, \mathcal Q^{(i);\pm}_E)$ for \emph{some $i$} by \eqref{8393hf93j3fojf83} which, by Definition \ref{kd93jf03j0j3j3p}, is called an obstruction class;
	\item if $\Xfr$ is non-split, then there will not exist a level-$(q+1)$ trivialisation for $\Xfr$, and;
	\item if $\Xfr$ is split, then there will exist a level-$(q+1)$ trivialisation for $\Xfr$.
\end{enumerate}
As remarked in \cite[p. 34]{DW1}, the obstruction classes themselves do not generally constitute a system of invariants of a given supermanifold. To elaborate further on this point we provide the following illustration.

\begin{ILL}\label{kopdj389hf883jfpo3}
Consider the split model $\Pi E$. By Theorem $\ref{dn8j83jp3k9k922}$ it will define a trivialisation $\rho$ in $H^1(M, \Gc^{(2)}_E)$, and by construction this trivialisation $\rho$ corresponds to the choice of base-point in the pointed set $H^1(M, \Gc^{(2)}_E)$. Furthermore, we also have that $\om_*(\rho) = 0$ in $H^1(M, \mathcal Q^{(2);+}_E)$ since $\om_*$ preserves the base-point. Then by exactness of the sequence in \eqref{8393hf93j3fojf83} at $i = 2$, we will obtain a lift of $\rho$ to some level-$3$ trivialisation $\vp_3\in H^1(M, \Gc^{(3)}_E)$. The map $H^1(M, \Gc^{(3)}_E) \ra H^1(M, \Gc^{(2)}_E)$ is, in general, not injective however so there may exist many lifts $\vp_3$ of $\rho$ and there is no reason to suppose, in general, that $\om_*(\vp_3) = 0$ in $H^1(M, \mathcal Q^{(3);-}_E)$. In accordance with \emph{(iii)} above, if $\Xfr$ is split then associated to $\rho$ will be lifts $\vp_i\in H^1(M, \Gc^{(i)}_E)$ such that $\om_*(\vp_i) = 0$ for each $i$. Lifts $\vp_i$ of the split supermanifold $\rho$ for which $\om_*(\vp_i)\neq0$ are referred to in \cite{DW1} as `exotic lifts'.
\end{ILL}

\begin{REM}
\emph{In terms of thickenings, an exotic structure on the split model $\Pi E$ is a supermanifold $\Xfr_{(M, E)}$ with the filtration
\begin{align}
M \subset \Pi E^{(1)}\subset \Pi E^{(2)}\subset \cdots \subset \Pi E^{(l-1)}\subset \Xfr_{(M, E)}^{(l)}\subset \cdots \subset \Xfr_{(M, E)}
\end{align}
where:
\begin{enumerate}[(i)]
	\item $2<l<q$;
	\item $\Xfr_{(M, E)}\cong \Pi E$ and; 
	\item $[\Xfr_{(M, E)}^{(l)}]\in \mathrm {\bf T}^1(\Pi E^{(l-1)})$ is non-vanishing. 
\end{enumerate}	
We refer to \cite{BETTHIGHOBS} where exotic structures on supermanifolds are studied in more detail.}
\end{REM}

\subsection{Supermanifolds and Thickenings}
We revisit thickenings now by considering the following variant of Question \ref{kd93f3jf3kf9p4}:

\begin{QUE}\label{dmopeififp3p33}
Let $\eta\in H^1(M,\mathcal Q_E^{(i);\pm})$. Then does there exist a supermanifold $\Xfr_{(M, E)}$ which admits a trivialisation whose obstruction class is realised by this given element $\eta$?
\end{QUE}

From Corollary \ref{dopmropvopejpv}, the space $H^1(M,\mathcal Q_E^{(i);\pm})$ can be identified with $i$-th order thickenings $\Xfr_{(M, E)}^{(i)}$ of $\Pi E^{(i-1)}$. In particular, a partial answer to Question \ref{dmopeififp3p33} is readily given by Proposition \ref{o3jdioj3ijp3}: \emph{if $\pt_*(\eta) \neq0$, then there will not exist any supermanifold realising $\eta$ as an obstruction class}. Hence we find in Proposition \ref{o3jdioj3ijp3} a necessary condition for $\eta$ to be realised by some supermanifold. Naturally, we can then ask: \emph{is this sufficient?} This question is addressed in the following result:

\begin{PROP}\label{ppoppieepeee}
$\img~\om_*\subset \ker\pt_*$ but not necessarily conversely.
\end{PROP}

\begin{proof}
By Corollary \ref{dopmropvopejpv} we know that $H^1(M,\mathcal Q_E^{(i);\pm})\cong {\bf T}^1(\Pi E^{(i-1)})$. Let $\Xfr^{(i)}_{(M, E)}\supset \Pi E^{(i-1)}$ be a thickening and suppose that $\pt_*([\Xfr^{(i)}_{(M, E)}]) = 0$. Then by Proposition \ref{o3jdioj3ijp3} it will be unobstructed. Hence we can find a thickening $\Xfr^{(i+1)}_{(M, E)}\supset \Xfr^{(i)}_{(M, E)}$. Now, by Theorem \ref{dn8j83jp3k9k922} and the succeeding discussion, we know that any supermanifold $\Xfr_{(M, E)}$ will define an obstruction class in $H^1(M, \mathcal Q_E^{(i);\pm})$ for \emph{some $i$}, not necessarily unique, but in the image of the obstruction map $\om_*$. Suppose it admits a level-$i$ trivialisation with $\eta = \om_*(\vp_i)$. Then we can write:
\begin{align}
M \subset \Pi E^{(i)} \subset \cdots\subset \Pi E^{(i-1)}\subset \Xfr^{(i)}_{(M, E)}\subset \Xfr^{(i+1)}_{(M, E)}\subset\cdots \subset \Xfr_{(M, E)}
\label{dkjfjf03j93322}
\end{align}
and evidently $\om_*(\vp_i) = [\Xfr_{(M, E)}^{(i)}]$. Importantly, we see that the thickening $\Xfr^{(i+1)}_{(M, E)}\supset \Xfr^{(i)}_{(M, E)}$ must be unobstructed and so must thickenings of it, and so on (c.f., Remark \ref{dnclncjkdddjdjd}). More generally however there is of course no reason to expect, just because the thickening $\Xfr^{(i)}_{(M, E)}\supset \Pi E^{(i-1)}$ is unobstructed, that the thickening $\Xfr^{(i+1)}_{(M, E)}\supset \Xfr^{(i)}_{(M, E)}$ will be unobstructed also. In this case we will not be able to construct a supermanifold as in \eqref{dkjfjf03j93322} and so we \emph{cannot} write $ [\Xfr_{(M, E)}^{(i)}] = \om_*(\vp_i)$ for any level-$i$ trivialisation $\vp_i$. This shows that while we certainly have $\img~\om_*\subset \ker\pt_*$, the reverse containment need not hold.  
\end{proof}

The above result motivates the following definition.

\begin{DEF}\label{f904fj04fjj33}\emph{A thickening that defines a non-trivial class in $\ker\pt_*/\img~\om_*$ is referred to as a \emph{pseudo-supermanifold}. }
\end{DEF}

From the argument presented in the proof of Proposition \ref{ppoppieepeee}, it is clear that pseudo-supermanifolds can be characterised as those thickenings $\Xfr^{(l)}_{(M, E)}$ given by a filtration:
\begin{align}
M \subset \Pi E^{(1)} \subset \cdots\subset \Pi E^{(l-1)} \subset \Xfr^{(l)}_{(M, E)} \subset \cdots \subset \Xfr^{(k)}_{(M, E)}
\label{4j903f03jf03j9}
\end{align}
for \emph{some} $k$, where: 
\begin{enumerate}[(1)]
	\item $l<k< q$, for $q$ the rank of the vector bundle $E$; 
	\item that $[\Xfr^{(l)}_{(M, E)}]\in \ker\pt_*/\img~\om_*$ is non-trivial and;
	\item $\pt_*([\Xfr^{(k)}_{(M, E)}])\neq0$. 
\end{enumerate} 
Hence, for any $i$, we have:
\begin{align}
H^1(M, \mathcal Q^{(i);\pm}_E) = \img~\om_*\oplus \left(\frac{\ker \pt_*}{\img~\om_*}\right)\oplus (\ker\pt_*)^\perp.
\label{49d93j8fj3fj3p}
\end{align}
In words, we can decompose the cohomology groups $H^1(M, \mathcal Q_E^{(i);\pm})$ into:
\begin{enumerate}[(i)]
	\item supermanifolds; 
	\item pseudo-supermanifolds and; 
	\item obstructed thickenings. 
\end{enumerate}
This addresses Question \ref{dmopeififp3p33}. In what follows we will illustrate some of the results so far presented in a more invariant manner for thickenings of even order.

\subsection{Thickenings of Even Order} 
We have appealed to trivialisations to prove all the results so far presented. In the present section we will give simple illustrations of Proposition \ref{o3jdioj3ijp3} and \ref{ppoppieepeee} for thickenings $\Xfr^{(2j)}_{(M, E)}\supset \Pi E^{(2j-1)}$. The new ingredient here is the following morphism of exact sequences concerning the groups $\Gc^{(i)}_E$. This result extends Proposition \ref{dm30j8fjf3o} and can be found in \cite[p. 55]{ONISHCLASS}.

\begin{PROP}\label{dmopifoijfojfo3}
There exists a morphism of short-exact sequences of sheaves of groups on $M$:
\[
\xymatrix{
1\ar[r] & \ar@{^{(}->}[d] \Gc^{(2j+2)}_E \ar@{^{(}->}[r] &\ar@{=}[d] \Gc^{(2j)}_E\ar[r]^{\lam~~} &\ar[d]^p \Tfr_{\Pi E}[2j]\ar[r] & 1
\\
1 \ar[r] & \Gc^{(2j+1)}_E \ar@{^{(}->}[r] & \Gc^{(2j)}_E \ar[r]^{\om~~} & \mathcal Q_E^{(2j);+}\ar[r] & 1
}
\]
\qed
\end{PROP}

By naturality of the sheaf cohomology functor $H^i$ we will obtain the following diagram on cohomology from the diagram in Proposition \ref{dmopifoijfojfo3} (horizontal) and the short-exact sequence in Lemma \ref{hjbbjhhveyjj} (vertical):
\begin{align}
\xymatrix{
& &\vdots\ar[d]
\\
& & H^1(M, \mathcal Q^{(2j+1);-}_E)\ar[d]
\\
\cdots \ar[r] & \ar@{=}[d]H^1(M, \Gc^{(2j)}_E) \ar[r]^{\lam_*~~} & H^1(M, \Tfr_{\Pi E}[2j])\ar[d]^{p_*}
\\
\cdots\ar[r] & H^1(M, \Gc^{(2j)}_E)\ar[r]^{\om_*} & H^1(M, \mathcal Q_E^{(2j);+})\ar[d]^{\pt_*}
\\
& & H^2(M, \mathcal Q_E^{(2j+1);-})\ar[d]
\\
& & \vdots
}
\label{miojfioejfoijf039}
\end{align}
From the above diagram we have the following:
\begin{enumerate}[(1)]
	\item the map $\pt_*$ in Proposition \ref{o3jdioj3ijp3} coincides with $\pt_* : H^1(\mathcal Q_E^{(2j);+}) \ra H^2(\mathcal Q_E^{(2j+1);-})$ in \eqref{miojfioejfoijf039} in the case where $\Xfr^{(l)}_{(M, E)} = \Pi E^{(2j-1)}$ (c.f., Corollary \ref{dopmropvopejpv}) and;
	\item commutativity of \eqref{miojfioejfoijf039} and long-exactness on cohomology allows us to deduce:
	\[
	\img~\om_* = \img~(p_*\circ \lam_*) \subset \img~p_* = \ker\pt_*.
	\]
	Hence we have $\img~\om_*\subset \ker\pt_*$, which is an instance of Proposition \ref{ppoppieepeee}. 
\end{enumerate}
We conclude now our foray into thickenings and supermanifolds. In the remaining sections we comment on the corresponding moduli problems and provide illustrations of obstructed thickenings of the complex projective plane. Before doing so however, we will compare the results obtained here with what one might find in the literature in what follows.

\subsection{Comparisons with Known Results}
The results so far presented and discussed in this article---mainly Proposition \ref{o3jdioj3ijp3} and Theorem \ref{o4993uf0oijwoijfwojpw}, are reminiscent of similar results one finds in studies of deformation theory in both complex-analytic and algebraic geometry. In the complex-analytic setting, we can compare Theorem \ref{jnvjknknvkjnvk} and Proposition \ref{o3jdioj3ijp3} with \cite[Theorem 5.1, p. 214]{KS}. In the algebraic setting, we compare Proposition \ref{o3jdioj3ijp3} with \cite[Tag 08L8]{STACKPROJ} and Theorem \ref{o4993uf0oijwoijfwojpw} with \cite[Tag 08UC]{STACKPROJ}. Where complex supermanifolds are concerned, the development of a deformation theory for them can be found in works such as \cite{ROTH} in the complex-analytic setting and \cite{VAIN} in the more algebro-geometric setting. 
\\\\
In \cite{ROTH}, the main idea is to regard a complex supermanifold as a deformation of its split model. A guiding question is then: \emph{what are the obstructions to deforming the split model?}  Sufficient conditions, being the vanishing of the second cohomology group, are identified in \cite[Theorem 3, p. 259]{ROTH}. In this way, the central result in \cite{ROTH} mirrors Theorem \ref{jnvjknknvkjnvk} and thereby also the results in \cite{EASTBRU}. In addition to these sufficient conditions, we consider necessary conditions in this article. Indeed the subject of the illustrations on the projective plane, provided at the end of this article, focus on these necessary conditions.
\\\\
The developments in \cite{VAIN} follow the more algebraic route to deformation theory, in the spirit of \cite{HARTDEF, STACKPROJ}. These developments are however concerned more with the deformation theory of a given supermanifold, as opposed to the problem of existence of a supermanifold as a deformation. In this subtle way the articles \cite{ROTH} and \cite{VAIN} differ\footnote{\label{ijfoeho8f4ho}this is explicitly mentioned by Vaintrob in \cite{VAIN} in the introductory paragraph. The difference is attributed to the difference between the terms \emph{classification} and \emph{deformation}.} and the material presented so far in this article is more closely related to \cite{ROTH}. In the following section we will consider the moduli problem for supermanifolds and thickenings.

\section{Moduli Problems: Supermanifolds and Thickenings}\label{6}

\noindent
Moduli problems in considerable generality are discussed in \cite[p. 150]{HARTDEF}. Loosely speaking, one aspect of the moduli problem concerns itself with giving sets such as \eqref{rkfj490fj4fp3okp3ofk3p} the structure of an algebraic variety or scheme or some other such geometric object. Such problems admit nice, categorical reformulations where they can be reduced to the problem of representability of an appropriately constructed functor. 
\\\\
In \cite{HARTDEF}, sets such as \eqref{rkfj490fj4fp3okp3ofk3p} are themselves referred to as moduli problems. In this section we will discuss the moduli problem for complex supermanifolds, following studies by Onishchik in \cite{ONISHMOD, ONISHCLASS}. We aim to then formulate a variant of this problem in order to investigate the structure of the corresponding moduli variety. Our objective will be to formulate some general conjectures on the structure of this variety. Finally, in taking motivation from the decomposition in \eqref{49d93j8fj3fj3p}, we conclude by arguing that this decomposition will hold at the level of moduli, i.e., up to isomorphism.

\subsection{Framed Supermanifolds} The moduli problem for complex supermanifolds, modelled on a given complex manifold $M$ and holomorphic vector bundle $E$, is written down in \eqref{rkfj490fj4fp3okp3ofk3p} and given further meaning in Theorem \ref{dn8j83jp3k9k922}. We consider here a variant of the moduli problem in \eqref{rkfj490fj4fp3okp3ofk3p}. 
\\\\
In Definition \ref{ifh4fgiuho48h04} we defined a coframe for a supermanifold $\Xfr_{(M, E)}$.  Dual to a coframe is a frame and we define it in what follows. Starting with a supermanifold  modelled on $(M, E)$, recall the exact sequences in Lemma \ref{jd90d098jdoijp3} and \ref{hjbbjhhveyjj}:
\begin{align}
0\ra \Tfr_{(M, E)}[0] \hookrightarrow \Tfr_{(M, E)}[-1] \twoheadrightarrow \Tfr_{\Pi E}[-1] \ra0.
\label{rofopjofj3o9fj930jf3}
\end{align}
Since $\Tfr_{(M, E)}[-1] = \Tfr_{(M, E)}$ and $\Tfr_{\Pi E}[-1] \cong \Ec^\vee$ we obtain from \eqref{rofopjofj3o9fj930jf3} a surjective morphism of sheaves $\al: \Tfr_{(M, E)} \twoheadrightarrow \Ec^\vee$ which is an isomorphism modulo $\Tfr_{(M, E)}[0]$. 

\begin{DEF}\label{fmojfiojf4fjp3}
\emph{To any supermanifold $\Xfr_{(M, E)}$, an epimorphism $\al:\Tfr_{(M, E)} \ra \Ec^\vee$ with  $\ker \al = \Tfr_{(M, E)}[0]$ will be referred to as a \emph{framing}. A supermanifold $\Xfr_{(M, E)}$ equipped with a choice of framing $\al$ will be called \emph{framed}}.
\end{DEF}

To justify the terminology `coframe' we have:

\begin{LEM}\label{fh389fh98hfoij83}
Let $\al$ denote a framing on $\Xfr_{(M,E)}$. Then $\al^\vee$ is a coframe modulo $\Tfr_{(M, E)}[0]$.
\end{LEM}

\begin{proof}
Recall from Definition \ref{ifh4fgiuho48h04} that a \emph{coframe} is an isomorphism between $\Jc/\Jc^2$ and $\Ec$. Dual to this is an isomorphism $(\Jc/\Jc^2)^\vee\cong \Ec^\vee$ which is precisely what we get from the data of a framing $\al$.
\end{proof}

The reason for introducing a framing in this section is so we can make sense of the following notion of equivalence.\footnote{We could also consider an equivalence by appealing directly to the coframe $\vp : \Jc/\Jc^2\stackrel{\cong}{\ra}\Ec$. However, it seems more natural to use framings. Moreover, framings are more closely related to similar discussions in \cite{DW1}.}

\begin{DEF}\label{rjf930j09j3fjp3fjp3w}\emph{Two framed supermanifolds $(\Xfr_{(M, E)}, \al)$ and $(\Xfr_{(M, E)}^\p, \al^\p)$ are said to be \emph{equivalent} if:
\begin{enumerate}[(i)]
	\item there exists an isomorphism $\lam: \Xfr_{(M, E)}\stackrel{\cong}{\ra}\Xfr_{(M, E)}^\p$ which is trivial when restricted to $\Pi E^{(1)}$ and;
	\item the isomorphism $\lam$ preserves the framing, i.e., that its differential $\lam_*$ induces a commutative diagram
	\[
	\xymatrix{
	\Tfr_{(M, E)}\ar[d]_{\lam_*} \ar[r]^{\al} & \Ec^\vee \\
	\Tfr_{(M, E)}^\p\ar[ur]_{\al^\p} &
	}
	\]
\end{enumerate}
where $\Tfr_{(M, E)}$, resp. $\Tfr_{(M, E)}^\p$ is the tangent sheaf of $\Xfr_{(M, E)}$, resp. $\Xfr_{(M, E)}^\p$. 
}
\end{DEF}

Similarly to \eqref{rkfj490fj4fp3okp3ofk3p}, the moduli problem here is:
\begin{align}
\Mfr_{(M, E)}^{\mathrm{framed}} = 
\left\{\mbox{framed supermanifolds $(\Xfr_{(M, E)}, \al)$}\right\}/\sim
\label{rkfj490fj4fp3okp3ofk3p9d039jd}
\end{align}
where the equivalence `$\sim$' in \eqref{rkfj490fj4fp3okp3ofk3p9d039jd} above is in the sense of Definition \ref{rjf930j09j3fjp3fjp3w}. As a set $\Mfr_{(M, E)}^{\mathrm{framed}}$ consists of equivalence classes of pairs $(\Xfr_{(M, E)}, \al)$ where $\al : \Tfr_{(M, E)} \twoheadrightarrow \Ec^\vee$ is a framing.

\begin{REM}\label{rfh784gf87h7fh398fh3}
\emph{The moduli problem $\Mfr_{(M, E)}^{\mathrm{framed}}$ is considered also in \cite[p. 2151]{VAIN}. In \cite[Chapter 3]{VAIN} it is argued that the corresponding moduli functor is quasi-representable. 
}
\end{REM}

 Evidently we have a map $\Mfr_{(M, E)}^{\mathrm{framed}} \ra \Mfr_{(M, E)}$ given by forgetting this choice of framing $\al$. The moduli problems \eqref{rkfj490fj4fp3okp3ofk3p} and \eqref{rkfj490fj4fp3okp3ofk3p9d039jd} are related in this way. Indeed, it is shown in the proof of Theorem \ref{dn8j83jp3k9k922} in \cite{GREEN} and remarked also in \cite[p. 32]{DW1} the following, which we state as a proposition and for which we provide a sketch of a proof:

\begin{PROP}\label{jf983f3jfoj3ojg93jgp}
The moduli problem for framed, complex supermanifolds modelled on a given pair $(M, E)$ is in one-to-one correspondence with $H^1(M, \Gc^{(2)}_E)$. 
\end{PROP}

\noindent
\emph{Proof Sketch}. 
A framed supermanifold is a supermanifold $\Xfr_{(M, E)}$ equipped with a framing $\al$. We are free to change $\al$ by a global automorphism of $\Ec$, i.e., an element of $H^0(M, \mathscr Aut~\Ec)$ and so $\al$ is certainly not unique.  Now suppose $\Xfr_{(M, E)}$ and $\Xfr_{(M, E)}^\p$ are framed with respective framings $\al$ and $\al^\p$ and moreover suppose $(\Xfr_{(M, E)}, \al)\sim (\Xfr^\p_{(M, E)}, \al^\p)$. Then there exists an isomorphism $\lam : \Xfr_{(M, E)}\stackrel{\cong}{\ra}\Xfr^\p_{(M,E)}$ characterised by Definition \ref{rjf930j09j3fjp3fjp3w}. By Definition \ref{rjf930j09j3fjp3fjp3w}(i) we have $\Oc_M/\Jc^2 = \Oc_M^\p/\Jc^{\p2}$. The isomorphism $\lam$ preserves the $\Zbb_2$-grading which means $\Jc/\Jc^2 = \Jc^\p/\Jc^{\p2}$.\footnote{c.f., the proof of Lemma \ref{fiugf78g79h380fj09}.} That the coframes $\vp$ and $\vp^\p$ for $\Xfr_{(M, E)}$ resp. $\Xfr_{(M, E)}^\p$ coincide follows from Definition \ref{rjf930j09j3fjp3fjp3w}(ii) and Lemma \ref{fh389fh98hfoij83}. Hence the frames $\al$ and $\al^\p$ coincide, modulo $\Tfr_{(M, E)}[0]$.  Thus $H^0(M, \mathscr Aut~\Ec)$ separates equivalence classes of framed supermanifolds. Using this and the fact that any supermanifold modelled on $(M, E)$ defines an element in $H^1(M, \Gc^{(2)}_E)$, we deduce $\Mfr_{(M, E)}^{\mathrm{framed}}\subseteq H^1(M, \Gc^{(2)}_E)$. The reverse inclusion $H^1(M, \Gc^{(2)}_E)\subseteq \Mfr_{(M, E)}^{\mathrm{framed}}$ follows from the construction of the (sheaf of) groups $\Gc^{(2)}_E$ in \eqref{fkfugunefefmelkfe}. We omit the details here.
\qed

\begin{REM}
\emph{Comparing Proposition \ref{jf983f3jfoj3ojg93jgp} with Theorem \ref{dn8j83jp3k9k922}, the map $\Mfr_{(M, E)}^{\mathrm{framed}} \ra \Mfr_{(M, E)}$ is given by the quotient map $H^1(\Gc^{(2)}_E) \ra H^1(\Gc^{(2)}_E)/H^0(\mathscr Aut~\Ec)$. In this way, we can think of the action of $H^0(M, \mathscr Aut~\Ec)$ on $H^1(M, \Gc^{(2)}_E)$ as changing the frame.}
\end{REM}

We have so far been embroiled in set-theoretic aspects of the moduli problem. Onishchik in \cite[p. 68]{ONISHCLASS} addresses the moduli problem in \eqref{rkfj490fj4fp3okp3ofk3p9d039jd} by constructing an algebraic variety parametrising framed supermanifolds. Paraphrasing this result we have:

\begin{THM}\label{4jf93jf903jf3jfp3j}
Let
\[
\Tfr_{\Pi E}^{\geq2} := \bigoplus_{k\geq 2} \Tfr_{\Pi E}[k].
\]
Then there exists a connected, complex-analytic subvariety $\mathbb V_{(M, E)}\subseteq H^1(M, \Tfr_{\Pi E}^{\geq2})$ whose set of points map onto $\Mfr^{\mathrm{framed}}_{(M, E)}$. \qed
\end{THM}

As briefly mentioned at the outset of this section, a resolution of the moduli problem consists of finding a geometric object parametrising the given set of isomorphism classes. Where supermanifolds are concerned, the results in \cite{VAIN} suggest that we might expect the corresponding moduli space to be a space in the sense of supergeometry, i.e., some sort of superspace. We will refer to such a space as a \emph{supermoduli space}.\footnote{At the beginning of this article, by way of motivation, the moduli space of super Riemann surfaces was mentioned. This object is an example of a supermoduli space.} For the super-geometric analogue of varieties and schemes see  \cite{LEITSPEC, RABGLOB, KAPSUSY, BETTVAR}. For framed supermanifolds $\Mfr_{(M, E)}^{\mathrm{framed}}$ it is tempting to conclude from Theorem \ref{4jf93jf903jf3jfp3j} that the corresponding supermoduli space is the variety $\mathbb V_{(M, E)}$. However, it is not clear as to how this variety can be interpreted as an object in supergeometry. That is, as a `super'-variety or `super'-scheme. One possible reason for this might be due to the distinction between \emph{classification} and \emph{deformation}, mentioned briefly in footnote \eqref{ijfoeho8f4ho}. Indeed, it is remarked in the introductory section in \cite{ONISHCLASS} that $\mathbb V_{(M, E)}$ is related to classifications, not deformations. Supermoduli spaces however should be related to deformations. As such, using the language in \cite{HARTDEF}, we conjecture the following:

\begin{CONJ}\label{fkrpokvpokvp4k}
There exists a coarse supermoduli space $\mathscr V_{(M, E)}$ representing\footnote{c.f., Remark \ref{rfh784gf87h7fh398fh3}.} the moduli problem $\Mfr_{(M, E)}^{\mathrm{framed}}$ with reduced space $(\mathscr V_{(M, E)})_{\mathrm{red}} = \mathbb V_{(M, E)}$. 
\end{CONJ}

In what follows we turn to another aspect of the moduli problem $\mathfrak M_{(M, E)}^{\mathrm{framed}}$, motivated by Proposition \ref{jf983f3jfoj3ojg93jgp}.

\subsection{The Moduli Problem with Level Structures}
We wish to formulate here a variant of the moduli problem in \eqref{rkfj490fj4fp3okp3ofk3p9d039jd} and submit a conjecture regarding the structure of the moduli variety $\mathbb V_{(M, E)}$. The variant we have in mind is motivated, in a sense, by the entire theme underlying this article---that of thickenings and filtrations. Recall that any complex supermanifold comes equipped with thickenings that fit together to define a filtration, depicted in \eqref{bcbfjhbvrbjh}. We suspect that a similar structure will manifest itself on the moduli variety $\mathbb V_{(M, E)}$ and in this section we will clarify our suspicions further. We begin with the following definitions:

\begin{DEF}\label{nfiuh78fg47fh93hf038}
\emph{A supermanifold is said to be be \emph{$j$-trivialisable} if it admits a level-$(j+1)$ trivialisation $\vp_j$. A supermanifold equipped with a level-$(j+1)$ trivialisation $\vp_j$ is said to be \emph{$j$-trivialised} and is denoted by the pair $(\Xfr_{(M, E)}, \vp_j)$.\footnote{c.f., Definition \ref{rfijf903jf903kf}.}
}
\end{DEF}

A $j$-trivialised supermanifold $\Xfr_{(M, E)}$ is filtered as follows:
\begin{align}
M \subset \Pi E^{(1)} \subset \Pi E^{(2)}\subset \cdots\subset \Pi E^{(j)}\subset \Xfr_{(M, E)}^{(j+1)}\subset \cdots\subset \Xfr_{(M, E)}.
\label{fiorjfiofojfp3}
\end{align}
Given a $j$-trivialised supermanifold $(\Xfr_{(M, E)}, \vp_j)$, an isomorphism $\Xfr_{(M, E)}\cong \Xfr_{(M, E)}^\p$ certainly need not preserve the level of the trivialisation $\vp_j$, as Illustration \ref{kopdj389hf883jfpo3} shows. As such we consider the following notion of equivalence:

\begin{DEF}\label{fiojfoj90fj93jpf3}\emph{Fix $j$ and let $(\Xfr_{(M, E)}, \vp_i)$ be an $i$-trivialised supermanifold for $i\geq j$. The \emph{$j$-equivalence class} of $(\Xfr_{(M, E)}, \vp_i)$ consists of level-$i^\p$ trivialised supermanifolds $(\Xfr^\p_{(M, E)}, \vp^\p_{i^\p})$ such that:
\begin{enumerate}[(i)]
	\item $i^\p\geq j$ and;
	\item $\Xfr_{(M, E)} \cong \Xfr^\p_{(M, E)}$.
\end{enumerate}
}
\end{DEF}

We present now the following variant of the moduli problem in \eqref{rkfj490fj4fp3okp3ofk3p9d039jd}:
\begin{align}
\Mfr_{(M, E)}^{(j);\mathrm{triv.}} = 
\left\{\mbox{level-$i$ trivialised supermanifolds with $i\geq j$}\right\}/\sim
\label{rkfj490fj3r3f3ff334fp3okp3ofk3p9d039jd}
\end{align}
the above equivalence being that in the sense of Definition \ref{fiojfoj90fj93jpf3}.

\begin{LEM}\label{fgdfvbhjbvevy3r4}
$\Mfr_{(M, E)}^{(1);\mathrm{triv.}} = \Mfr_{(M, E)}^{\mathrm{framed}}$.
\end{LEM}

\begin{proof}
A 1-trivialisation for a supermanifold $\Xfr_{(M, E)}$ is a trivialisation $(\Ufr, \rho)$ described in Construction \ref{knkniucbruibckj}, where $\rho$ is as in \eqref{uicuibviuv78v3} and \eqref{porviorhg894h89}. If $(x^\mu, \q_a)$ denote coordinates on a patch $\Uc$ a vector field $X$ can be written
\[
X = f^\mu\frac{\pt}{\pt x^\mu} +  g_a\frac{\pt}{\pt \q_a}
\]
where the indices are implicitly summed. Now consider a map $\al$ sending $X$ to $(g_a \mod \Jc)~\pt/\pt\q_a$ and observe that on the intersection $\Uc\cap\Vc$,
\begin{align*}
g_a\frac{\pt}{\pt \q_a} 
&= g_a \circ f_{UV} \left(\frac{\pt \rho^\nu_{\Uc\Vc}}{\pt \q_a}\frac{\pt}{\pt y^\nu}
+
\frac{\pt \rho_{\Uc\Vc;b}}{\pt \q_a}\frac{\pt}{\pt \eta_b}
\right)
\stackrel{\al}{\longmapsto}
\left(g_a \circ f_{UV} \mod \Jc\right) \zeta_{UV, b}^a\frac{\pt}{\pt \eta_b}
\end{align*}
where $(y^\nu,\eta_b)$ denote coordinates on $\Vc$. Hence 
\[
(g_a\mod \Jc)~\frac{\pt}{\pt \q_a} 
=
\left(g_a \circ f_{UV} \mod \Jc\right) \zeta_{UV, b}^a\frac{\pt}{\pt \eta_b}
\]
and so $\al(X)$ is a section of $\Ec^\vee$. By construction $\al: \Tfr_{(M, E)}\ra \Ec^\vee$ will be a framing. Note that the framing $\al$ depends on the trivialisation $(\Ufr, \rho)$ up to common refinement. Now by Definition \ref{fiojfoj90fj93jpf3}, two 1-trivialised supermanifolds $\Xfr_{(M, E)}$ and $\Xfr^\p_{(M, E)}$ will be equivalent if and only if there exists an isomorphism $\lam : \Xfr_{(M, E)}\stackrel{\cong}{\ra} \Xfr^\p_{(M, E)}$ that restricts to the identity along $\Pi E^{(1)}$. Comparing with Definition \ref{rjf930j09j3fjp3fjp3w}(i) it is not too hard to see that $\lam$ will also preserve the framing and thereby be an equivalence of framed supermanifolds. The converse assertion is straightforward.
\end{proof}

Following on from Lemma \ref{fgdfvbhjbvevy3r4} above we have the following generalisation of Proposition \ref{jf983f3jfoj3ojg93jgp}.

\begin{PROP}\label{dop3oiehoijip}
There exists a bijection: $\Mfr_{(M, E)}^{(j);\mathrm{triv.}}\cong H^1(M, \Gc^{(j+1)}_E)$.
\end{PROP}

\begin{proof}
The existence of map $\Mfr_{(M, E)}^{(j);\mathrm{triv.}}\ra H^1(M, \Gc^{(j+1)}_E)$ follows immediately from Definition \ref{nfiuh78fg47fh93hf038}. It is given by sending $(\Xfr_{(M, E)}, \vp_i) \mapsto \vp_i$, where $i\geq j$. That this map is a bijection is precisely what is captured by Definition \ref{fiojfoj90fj93jpf3}. 
\end{proof}

We tautologically have maps  $\Mfr_{(M, E)}^{(j);\mathrm{triv.}}\ra  \Mfr_{(M, E)}^{(j-1);\mathrm{triv.}}$ and forgetting the trivialisation corresponds to the map $\Mfr_{(M, E)}^{(j)}\ra \Mfr_{(M, E)}$. Evidently we have a commutative diagram:
\begin{align}
\xymatrix{
\ar[drrrr]\Mfr_{(M, E)}^{(j);\mathrm{triv.}} \ar[rr]& & \ar[drr]\Mfr_{(M, E)}^{(j-1);\mathrm{triv.}}\ar[r] & \cdots \ar[r]  & \Mfr_{(M, E)}^{\mathrm{framed}}\ar[d]
\\
&&&&\Mfr_{(M, E)}.
}
\label{4di3j8j30j3pkp3}
\end{align}
By Proposition \ref{dop3oiehoijip} we can compare the horizontal maps in \eqref{4di3j8j30j3pkp3} with the maps on 1-cohomology induced by the normal filtration of the sheaf of groups $\Gc_E^{(2)}$. Now as observed  in Illustration \ref{kopdj389hf883jfpo3}, the horizontal maps in \eqref{4di3j8j30j3pkp3} need not be injective and so the moduli variety parametrising them, should it exist, need not fit together to define a filtration of $\mathbb V_{(M, E)}$. With this observation and Theorem \ref{4jf93jf903jf3jfp3j} we conjecture:

\begin{CONJ}\label{dnohf3fh839hf3h}
Fix a complex manifold $M$ and holomorphic vector bundle $E\ra M$ of rank $q$ and set 
\[
\Tfr_{\Pi E}^{\geq j} := \bigoplus_{k\geq j} \Tfr_{\Pi E}[k].
\]
For each $j$ there exists a moduli problem $\widetilde \Mfr_{(M, E)}^{(j)}\subseteq \Mfr_{(M, E)}^{(j);\mathrm{triv.}}$ and a connected, algebraic subvariety $\mathbb V^{(j)}_{(M, E)}\subset  H^1(M,\Tfr_{\Pi E}^{\geq j})$ such that:
\begin{enumerate}[(i)]
	\item $\widetilde \Mfr_{(M, E)}^{(1)} = \Mfr_{(M, E)}^{\mathrm{framed}}$ and $\mathbb V^{(1)}_{(M, E)} = \mathbb V_{(M, E)}$;
	\item the set of points of $\mathbb V^{(j)}_{(M, E)}$ map onto $\widetilde \Mfr_{(M, E)}^{(j)}$ and;
	\item there exists a descending filtration 
	\[
	\mathbb V_{(M, E)} = \mathbb V_{(M, E)}^{(1)} \supset \mathbb V_{(M, E)}^{(2)}\supset \mathbb V_{(M, E)}^{(3)}\supset\cdots\supset \mathbb V_{(M, E)}^{(q)}.
	\]
\end{enumerate}
\end{CONJ}
~\\

\subsection{The Moduli Problem for Thickenings}
In Theorem \ref{4jf93jf903jf3jfp3j} it is the first-cohomology group of the sheaf $\Tfr_{\Pi E}^{\geq2}$ which is relevant. Note however that the obstruction spaces in \eqref{49d93j8fj3fj3p} incorporate the sheaves $\mathcal Q_E^{(j);\pm}$ and, while they are related to $\Tfr_{\Pi E}[j]$ via the short-exact sequences in \eqref{oinriobuoioie}, the sheaves $\mathcal Q_E^{(j);\pm}$ themselves contain important information (see e.g., (2) following \eqref{miojfioejfoijf039}). The first cohomology group of $\mathcal Q_E^{(j);\pm}$ comprises supermanifolds in addition to pseudo-supermanifolds and obstructed thickenings and so we suspect that it will house a variety representing, in a sense, a larger moduli problem than framed supermanifolds in \eqref{rkfj490fj4fp3okp3ofk3p9d039jd}. We will concentrate here on the moduli problem for thickenings directly and refrain from making any statements about the moduli variety itself. 
\\\\
Our intent in this section is to show that the decomposition in \eqref{49d93j8fj3fj3p} holds at the level of moduli, i.e., up to isomorphism. This leads to the following notion of a morphism of thickenings, generalising Definition \ref{equivextensusythick}. 

\begin{DEF}\label{303hfjeipjej3fj}
\emph{Let $\Xfr_{(M, E)}^{(k)} = (M, \Oc_M^{(k)})$ and $\Xfr_{(M, E)}^{(l)\p}= (M, \Oc_M^{(l)\p})$ be thickenings of a given pair $(M, E)$. A \emph{morphism} $\lam: \Xfr_{(M, E)}^{(k)}\ra \Xfr_{(M, E)}^{(l)\p}$ is defined to be a morphism of locally ringed spaces that:
\begin{enumerate}[(i)]
	\item preserves the $\Zbb_2$-grading, i.e., $\lam^\sharp (\Oc_M^{(l)\p;\mathrm{ev/odd}})\subset \Oc_M^{(k);\mathrm{ev/odd}}$, where $\Oc_M^{(k);\mathrm{ev/odd}}$ (resp. $\Oc_M^{(l)\p;\mathrm{ev/odd}}$) is the even and odd graded component of $\Oc_M^{(k)}$ (resp. $\Oc_M^{(l)\p})$ and;
	\item $\lam|_{\Pi E^{(1)}}$ is the identity, i.e., that the following diagram commutes:
	\[
	\xymatrix{
	\ar@{^{(}->}[d] \Pi E^{(1)} \ar@{^{(}->}[r] & \Xfr^{(k)}\ar[dl]^\lam
	\\
	\Xfr^{(l)\p} &
	}
	\]
\end{enumerate}
If $k =l$ and $\lam$ is invertible, then it is an isomorphism. 
}
\end{DEF}

The following construct, which we term the \emph{moduli problem for order-$k$ thickenings}, now makes sense:
\begin{align}
\mathcal M^{(k)}_{(M, E)} = \left\{
\mbox{thickenings of $(M, E)$ of order $k$}
\right\}/\cong.
\label{dk3j903jf093jf903j}
\end{align}
As mentioned above, we wish to argue that the decomposition in \eqref{49d93j8fj3fj3p} classifying thickenings into obstructed thickenings, pseudo-supermanifolds and supermanifolds, holds at the level of moduli. Our stating point is the following: 

\begin{PROP}\label{4f4f8j0fj0}
Let $\Xfr^{(k)}_{(M, E)}$ be a thickening of order $k$ and suppose that $\Xfr^{(k)}_{(M, E)}\cong \Xfr^{(k)\p}_{(M, E)}$, for another order-$k$ thickening $\Xfr^{(k)\p}_{(M, E)}$. Then there exists a bijection: 
\[
\mathrm{\bf T}^1(\Xfr_{(M, E)}^{(k)})\cong\mathrm{\bf T}^1(\Xfr_{(M, E)}^{(k)\p}).
\]
\end{PROP}

\begin{proof}
It suffices to show that if $\Xfr^{(k+1)}_{(M, E)}\supset \Xfr_{(M, E)}^{(k)}$ is a thickening, then we can use the given isomorphism $\Xfr^{(k)}_{(M, E)}\cong \Xfr^{(k)\p}_{(M, E)}$ to deduce that $\Xfr^{(k)\p}_{(M, E)}$ must be unobstructed. The method of proof is similar to that for Proposition \ref{o3jdioj3ijp3}. Firstly let $\Xfr^{(k+1)}_{(M, E)}\supset \Xfr_{(M, E)}^{(k)}$ be a thickening and denote by $(\Ufr^{(k+1)}, \rho^{(k+1)})$ a trivialisation for it. From $\rho^{(k+1)}$ we obtain the trivialisation $(\Ufr^{(k)}, \rho^{(k)})$ for $\Xfr_{(M, E)}^{(k)}$ by setting $\rho^{(k)} := \rho^{(k+1)}\mod \Jc^{k+1}$. Now, by assumption there exists an isomorphism $\lam: \Xfr^{(k)}_{(M, E)}\cong \Xfr^{(k)\p}_{(M, E)}$. If $(\Ufr^{(k)}, \rho^{(k)\p})$ denotes a trivialisation for $\Xfr_{(M, E)}^{(k)\p}$, then on the intersection $\Uc\cap\Vc$ we have:
\begin{align}
\lam_\Vc\circ \rho^{(k)}_{\Uc\Vc} = \rho^{(k)\p}_{\Uc\Vc}\circ \lam_\Uc.
\label{fiofio3jfio3jf3}
\end{align}
Before we proceed, the following construct will be convenient:
\begin{align}
\mathcal Q_E^{\geq2;\pm} := \bigoplus_{j\geq 2}\mathcal Q_E^{(j);\pm}.
\label{39f0j30fj3pjfp3vv}
\end{align}
Now just as in \eqref{fj4j49jf94kf4kp} we will find  that, more generally:
\begin{align}
\rho_{\Uc\Vc}^{(k)} - \rho^{(k)\p}_{\Uc\Vc} = (\dt\nu)_{\Uc\Vc} + w_{\Uc\Vc}
\label{4f4jf4jg4ogj49}
\end{align}
for $\nu = \{\nu_\Uc\}\in C^0(\mathfrak U, \mathcal Q_E^{(k);\pm})$ and $w = \{w_{\Uc\Vc}\}\in C^1(\Ufr, \mathcal Q_E^{\geq2;\pm})$.\footnote{note of course that \eqref{4f4jf4jg4ogj49} need not vanish modulo $\Jc^{k+2}$, c.f., the construction of the obstruction element $\Gam$ in \eqref{mchrehehruei38deu383j}.} The reason $w\neq 0$ here, in contrast to \eqref{fj4j49jf94kf4kp}, is for the reason that the isomorphism $\lam$ need not be trivial modulo $\Jc^k$. However, it must be trivial modulo $\Jc^2$ by Definition \ref{303hfjeipjej3fj}(ii), thereby justifying the summation in \eqref{39f0j30fj3pjfp3vv}. As a consequence of \eqref{4f4jf4jg4ogj49} we find:
\begin{align}
\Gam^{(k+1)\p}_{\Uc\Vc\Wc} &=  \rho^{(k)\p}_{\Uc\Wc} -   \rho^{(k)\p}_{\Vc\Wc}\circ \rho^{(k)\p}_{\Uc\Vc}
\notag\\
&= 
\rho^{(k)}_{\Uc\Wc} - (\dt\nu)_{\Uc\Wc} + w_{\Uc\Wc}
-
\rho^{(k)}_{\Vc\Wc}\circ \rho^{(k)}_{\Uc\Vc} - (\dt\nu)_{\Vc\Wc} - w_{\Vc\Wc} -(\dt\nu)_{\Uc\Vc} -  w_{\Uc\Vc}
\notag\\
&= 
\Gam_{\Uc\Vc\Wc}^{(k+1)} + (\dt w)_{\Uc\Vc\Wc}.
\label{38893hf83hf83h}
\end{align}
Note that while $w = \{w_{\Uc\Vc}\}\in C^1(\Ufr, \mathcal Q_E^{\geq2;\pm})$, the quantity $(\dt w)_{\Uc\Vc\Wc}$
 in \eqref{38893hf83hf83h} above will be homogeneous and of degree-$(k+1)$. This is because, by assumption, both $\Gam^{(k+1)\p}_{\Uc\Vc\Wc}$ and $\Gam_{\Uc\Vc\Wc}^{(k+1)}$ will vanish identically modulo $\Jc^{k+1}$ implying therefore that so must $(\dt w)_{\Uc\Vc\Wc}$. Hence, we can make sense of \eqref{38893hf83hf83h} as a statement about $\mathcal Q_E^{(k+1);\pm}$-valued 2-cocycles. In particular that $\{\Gam^{(k+1)\p}_{\Uc\Vc\Wc}\} = \Gam^{(k+1)\p} \sim \Gam^{(k+1)} = \{\Gam_{\Uc\Vc\Wc}^{(k+1)}\}$. Since we assume that $\Xfr^{(k)}_{(M, E)}$ is unobstructed we know that $\Gam^{(k+1)}\sim 0$ and therefore $\Gam^{(k+1)\p}\sim 0$ also which means $\Xfr^{(k)\p}$ will be unobstructed. The present result now follows from Theorem \ref{o4993uf0oijwoijfwojpw}.
\end{proof}

\begin{COR}\label{38f03hfoinio30}
An obstructed thickening cannot be isomorphic to an unobstructed thickening.
\end{COR}

\begin{proof}
If $\Xfr^{(k)}$ is an obstructed thickening then $\mathrm T^1(\Xfr^{(k)}) = \eset$ whereas for $\Xfr^{(k)\p}$ unobstructed, we have $\mathrm T^1(\Xfr^{(k)\p})\neq\eset$ by Theorem \ref{o4993uf0oijwoijfwojpw}. This corollary now follows from Proposition \ref{4f4f8j0fj0}. 
\end{proof}

Continuing on with the theme set by Corollary \ref{38f03hfoinio30}, we now consider pseudo-supermanifolds.

\begin{LEM}\label{fffggfefgjee}
Let $\Xfr^{(k+1)}\supset \Xfr^{(k)}$ and suppose $\Xfr^{(k)}\cong\Xfr^{(k)\p}$ for some other thickening $\Xfr^{(k)}$. Then there exists a thickening $\Xfr^{(k+1)\p}\supset \Xfr^{(k)\p}$ such that $\Xfr^{(k+1)\p}\cong  \Xfr^{(k+1)\p}$ as locally ringed spaces.
\end{LEM}

\begin{proof}
We are given the data of a thickening $\Xfr^{(k+1)}\supset \Xfr^{(k)}$ and an isomorphism $\Xfr^{(k)}\cong \Xfr^{(k)\p}$. With respect to a cover $\Ufr$, denote by $\rho^{(k+1)}$ and $\rho^{(k)\p}$ trivialisations for the thickenings $\Xfr^{(k+1)}$ and $\Xfr^{(k)\p}$ respectively. Note that $\rho^{(k)} := \rho^{(k+1)}\mod \Jc^{k+1}$ will be a trivialisation for $\Xfr^{(k)}$. Now just as in \eqref{fiofio3jfio3jf3}, the assumption that $\Xfr^{(k)}\cong \Xfr^{(k)\p}$ means there exists an isomorphism $\lam = \{\lam_\Uc\}$ such that on all intersections $\Uc\cap\Vc$,
\begin{align}
\lam_\Vc\circ \rho^{(k)}_{\Uc\Vc} = \rho^{(k)\p}_{\Uc\Vc}\circ \lam_\Uc.
\label{orooievio3o33}
\end{align}
Consider now the following object, defined on intersections:
\begin{align}
\rho^{(k+1)\p}_{\Uc\Vc} := \lam_\Vc\circ \rho^{(k+1)}_{\Uc\Vc}\circ \lam_\Uc^{-1}.
\label{39f03foeij40j3}
\end{align}
Since $\rho^{(k+1)}$ satisfies the cocycle condition, then so will $\rho^{(k+1)\p}$. Hence $(\Ufr, \rho^{(k+1)\p})$ will trivialise \emph{something}. We claim that it will define a trivialisation for a thickening $\Xfr^{(k+1)\p}\supset\Xfr^{(k)\p}$. The thickening $\Xfr^{(k+1)\p}$ will, by construction, be isomorphic to $\Xfr^{(k+1)}$ and the lemma will then follow. So to argue that $(\Ufr, \rho^{(k+1)\p})$ will trivialise $\Xfr^{(k+1)\p}\supset \Xfr^{(k)\p}$, it suffices to show that $\rho^{(k+1)\p} \equiv \rho^{(k)\p}$ modulo $\Jc^{k+1}$. To show this, write: $\rho^{(k+1)} = \rho^{(k)} + \phi^{(k+1)}$, for $\phi^{(k+1)}$ a $\mathcal Q_E^{(k+1);\pm}$-valued 1-cochain. This  will vanish modulo $\Jc^{k+1}$. Evaluating \eqref{39f03foeij40j3} gives:
\begin{align*}
\rho^{(k+1)\p}_{\Uc\Vc}
=
\lam_\Vc\circ \rho^{(k+1)}_{\Uc\Vc}\circ \lam_\Uc^{-1}
&= 
\lam_\Vc\circ 
\left(\rho_{\Uc\Vc}^{(k)} + \phi_{\Uc\Vc}^{(k+1)}\right)
\circ \lam_\Uc^{-1}
\\
&=
\lam_\Vc\circ \rho^{(k)}_{\Uc\Vc}\circ \lam_\Uc^{-1}
+\ldots
\\
&= \rho_{\Uc\Vc}^{(k)\p} + \ldots &&\mbox{by \eqref{orooievio3o33}}
\end{align*}
where the ellipses `$\ldots$' denote terms proportional to $\phi^{(k+1)}$. In particular, such terms lie in $\Jc^{k+1}$ which means, by \eqref{39f03foeij40j3}, that $\rho^{(k+1)\p} \equiv \rho^{(k)\p}$ modulo $\Jc^{k+1}$. The lemma now follows.
\end{proof}

As an illustration, we can use Proposition \ref{4f4f8j0fj0} and Lemma \ref{fffggfefgjee} to address Question \ref{kd93f3jf3kf9p4}:

\begin{LEM}\label{39f3hfoi33j}
Let $\Xfr^{(k)}$ be a thickening and suppose it is isomorphic to some thickening $\tilde\Xfr^{(k)}$ which embeds in some supermanifold $\tilde\Xfr$. Then there exists a supermanifold $\Xfr$ containing this given thickening $\Xfr^{(k)}$. 
\end{LEM}

\begin{proof}
We are given $\Xfr^{(k)}$. Since we assume $\Xfr^{(k)}\cong \tilde\Xfr^{(k)}$ and since $\tilde\Xfr^{(k)}$ is assumed to be embed in some supermanifold $\tilde\Xfr$, it follows that there will be a thickening $\tilde\Xfr^{(k+1)}\supset \tilde\Xfr^{(k)}$. Then by Lemma \ref{fffggfefgjee} we see that there must then exist a thickening $\Xfr^{(k+1)}\supset \Xfr^{(k)}$ with $\Xfr^{(k+1)}\cong \tilde\Xfr^{(k+1)}$. Now observe that we are in the same situation as we were at the start of this proof, i.e., that we have a thickening $\Xfr^{(k+1)}$ and an isomorphism between this given thickening and some thickening $\tilde\Xfr^{(k+1)}$ embedded in a supermanifold $\tilde\Xfr$. Continuing on inductively we can deduce (in finitely many steps) the existence of a supermanifold $\Xfr$, which is isomorphic to $\tilde\Xfr$, and contains $\Xfr^{(k)}$ as an embedded, $k$-th order thickening. 
\end{proof}

Now just like Corollary \ref{38f03hfoinio30} we have:

\begin{COR}\label{38hfofjoijfoijo}
A pseudo-supermanifold will not be isomorphic to any thickening that embeds in some supermanifold.\qed
\end{COR}

To recap now, from Corollary \ref{38f03hfoinio30} an obstructed thickening of order $k$ cannot be isomorphic to a pseudo-supermanifold of order-$k$; and from Corollary \ref{38hfofjoijfoijo} a pseudo-supermanifold can never be isomorphic to a thickening embedded in some supermanifold. Hence, a straightforward consequence is the following decomposition of the moduli problem, mirroring \eqref{49d93j8fj3fj3p}: 
\begin{align}
\mathcal M^{(k)}_{(M, E)} = \Mcl^{(k);\mathrm{smfld.}}_{(M, E)}\cup\mathcal M^{(k);\mathrm{pseudo.}}_{(M, E)} \cup \mathcal M^{(k);\mathrm{obs.}}_{(M, E)} 
\label{jifjefiojoif3j9f03j}
\end{align}
where $\Mcl^{(k);\mathrm{smfld.}}_{(M, E)}$; resp. $\mathcal M^{(k);\mathrm{pseudo.}}_{(M, E)}$; resp. $\mathcal M^{(k);\mathrm{obs.}}_{(M, E)}$ are moduli problems for order-$k$ thickenings embedded in supermanifolds; resp.  pseudo-supermanifolds of order $k$; resp. obstructed, $k$-th order thickenings. 
\\\\
In what follows we will infer the existence of obstructed thickenings of the complex projective plane.

\section{Illustrations over the Complex Projective Plane}\label{7}

\noindent
A corollary of Theorem \ref{jnvjknknvkjnvk} is that \emph{any} thickening $\Xfr^{(l)}$ of a Riemann surface $C$ will be unobstructed, in the sense of Definition \ref{smeioiof894f89jcoi}. We wish to describe here an example of an \emph{obstructed}, second order thickening so we will therefore have to look at thickenings $\Xfr^{(l)}_{(M, E)}$ of $M$ where $M$ a complex manifold with $\dim_\Cbb M\geq 2$ and $E\ra M$ is a holomorphic vector bundle of rank at least $3$. We will consider here the complex projective plane $M = \Cbb\Pbb^2$ and appeal to Proposition \ref{popiioh4o8fh84o}. Integral to our considerations will be the \emph{Bott formula}---a formula for computing the complex dimensions of certain cohomology groups on projective spaces; and \emph{Serre duality}. Regarding the Bott formula, we state this from \cite[p. 4]{OSS} below:
\begin{align}
h^q\left(\Om_{\Cbb\Pbb^n}^p(k)\right)
=
\left\{
\begin{array}{cl}
\binom{k+n-p}{k}\binom{k-1}{p} & \mbox{for $q = 0$ and $0\leq p\leq n$ and $k> p$}\\
1 & \mbox{for $0\leq p = q\leq n$ and $k = 0$}\\
\binom{-k+p}{-k}\binom{-k-1}{n-p} & \mbox{for $q = n$ and $0\leq p \leq n$ and $k<p-n$}\\
0 & \mbox{otherwise}.
\end{array}
\right.
\label{BOTT}
\end{align}
Here $\Om^p_{\Cbb\Pbb^n}(k) = \Om^p_{\Cbb\Pbb^n}\otimes\Oc_{\Cbb\Pbb^n}(k)$, where $\Oc_{\Cbb\Pbb^n}(k)$ is the $|k|$-th tensor power of the hyperplane divisor if $k> 0$, or its dual for $k< 0$; and $h^q( -) = \dim_\Cbb H^q(\Cbb\Pbb^n, -)$. The celebrated Serre duality theorem for vector bundles on projective space is:

\begin{THMSD}
\[
h^i(\Ec) = h^{n-i}(\Ec^\vee(-n-1))
\]
 for $E\ra \Cbb\Pbb^n$ a holomorphic vector bundle and $i = 0, \ldots, n$. \qed
\end{THMSD}

As mentioned, we are interested in constructing an example of an obstructed, second order thickening. Our starting point is the exact sequence in \eqref{jkkbkbkbk} for $k = 2$ which we give below for convenience:
\[
0 \ra \wedge^3\Ec\otimes\Ec^\vee\hookrightarrow \Tfr_{\Pi E}[2] \twoheadrightarrow \wedge^2\Ec\otimes\Tfr_M\ra0.
\]
This sequence induces a long-exact sequence on sheaf cohomology containing the following piece:
\begin{align}
\xymatrix{
\ldots\ar[r] & H^1(M,\wedge^3\Ec\otimes\Ec^\vee)\ar[r] & H^1(M,\Tfr_{\Pi E}[2]) \ar[r]^{r_*~~~} & H^1(M,\wedge^2\Ec\otimes\Tfr_M) \ar[d]^{\pt_*} \\ 
& &\ldots & \ar[l]H^2(M,\wedge^3\Ec\otimes\Ec^\vee)
}
\label{dmiomdiorjf94j09}
\end{align}
In the case where $\dim_\Cbb M=2$ we have:

\begin{LEM}\label{dmoeiciojciojc904jc409dmoeiciojciojc904jc409}
Let $M$ be a $2$-dimensional, complex manifold and $E\ra M$ a rank $3$, holomorphic vector bundle with $\deg E = k$. Suppose:
\begin{align}
h^1( \Tfr_{M}\otimes\wedge^2\Ec )\neq0;
&&
h^2(\Tfr_{M}\otimes\wedge^2\Ec )\neq0,
&&\mbox{and}&&
h^2( \Ec^\vee(k))\neq0,
\label{dmoeiciojciojc904jc409}
\end{align}
where $h^i(-) = h^i(M, -)$. Then there will exist an obstructed, second order thickening $\Xfr_{(M, E)}^{(2)}$ of $\Pi E^{(1)}$. \qed 
\end{LEM}

\begin{proof}
As $M$ is a two-dimensional, complex manifold it follows that $H^i(M, \Fc) = 0$ for all $i> 2$ and any sheaf of abelian groups $\Fc$. This allows us to conclude that the sequence in \eqref{dmiomdiorjf94j09} for $M$ continues as follows:
\begin{align}
\xymatrix{
\ldots\ar[r] & H^1(\Ec^\vee(k))\ar[r]^{s^1_*} & H^1(\Tfr_{\Pi E}[2]) \ar[r]^{r^1_*~~~} & H^1(\wedge^2\Ec\otimes\Tfr_M) \ar[d]^{\pt_*} \\ 
0&\ar[l]H^2(\wedge^2\Ec\otimes\Tfr_M) &\ar[l]_{~~~~~~r^2_*}H^2(\Tfr_{\Pi E}[2]) & \ar[l]_{~~s^2_*}H^2(\Ec^\vee(k))
}
\label{dmiomdiorjf94j043433d9}
\end{align}
where $H^i(-) = H^i(M, -)$. Our objective is to definitively conclude that the boundary map $\pt_*$ is non-trivial under the hypotheses in \eqref{dmoeiciojciojc904jc409}, for then we can conclude that the complement of $\img~r^1_*$ in $H^1(\wedge^2\Ec\otimes\Tfr_M)$, denoted $(\img~r^1_*)^\perp$, will be non-empty. This is important since, by exactness of \eqref{dmiomdiorjf94j043433d9}, Lemma \ref{lkliioolkss} and Theorem \ref{jnvjknknvkjnvk}, the set $(\img~r^1_*)^\perp$ will contain trivialisations for obstructed, second order thickenings. Now, the former two conditions in \eqref{dmoeiciojciojc904jc409} are clearly necessary conditions. To see that all three conditions in \eqref{dmoeiciojciojc904jc409} are sufficient to deduce $(\img~r^1_*)^c\neq\eset$, firstly note by exactness of \eqref{dmiomdiorjf94j043433d9} that: $\img~ \pt_* = \ker~s_*^2$. Hence it suffices to characterise the map $s^2_*$. Now, by exactness again we know that $r^2_*$ will be surjective and non-trivial. If $r^2_*$ is \emph{bijective}, then:
\begin{align}
\{0\} = \ker r^2_* = \img~s^2_* \LRa \ker~s_*^2\neq0 \LRa \img~\pt_*\neq0.
\label{ijoijoij494kk39kpk}
\end{align}
Hence, if $r^2_*$ is bijective, we see that $\pt_*$ will be non-trivial. Suppose however $r_*^2$ is \emph{not} bijective. Then, if $s^2_*$ is not bijective either, we may use the reasoning in \eqref{ijoijoij494kk39kpk} again to deduce that $\pt_*$ is non-trivial. However, if $s_*^2$ \emph{is} bijective, then $\pt_*$ must necessarily be trivial, but note that bijectivity of $s^2_*$ will contradict the surjectivity of $r^2_*$. Hence $s^2_*$ cannot be bijective, which means $\pt_*$ will be non-trivial. The lemma now follows. 
\end{proof}

\begin{EX}\label{8jf48f4unovu4no4oi}
(Split, rank $3$ vector bundles on $\Cbb\Pbb^2$) From \eqref{BOTT} we have:
\begin{align}
h^0(\Tfr_{\Cbb\Pbb^2}(l))&\neq 0 \iff l>2
\label{splitvdcp3condn1}
\\
h^1(\Tfr_{\Cbb\Pbb^2}(l)) &\neq 0\iff l=-3
\label{splitvdcp3condn2}
\\
h^2(\Oc_{\Cbb\Pbb^2}(l))&\neq 0\iff l<-3.
\label{splitvdcp3condn3}
\end{align}
Now let $\Ec = \bigoplus_{a=1}^3\Oc_{\Cbb\Pbb^2}(k_a)$ be a split, rank $3$, holomorphic vector bundle on $\Cbb\Pbb^2$.\footnote{Recall that a vector bundle is split if it can be written as a sum of line bundles. Now on any projective space $\Cbb\Pbb^n$, any holomorphic line bundle will be of the form $\Oc_{\Cbb\Pbb^n}(k)$ for some integer $k\in \Zbb$ (see, e.g., \cite[p. 145]{HARTALG}). Hence if $E\ra\Cbb\Pbb^2$ is split and of rank $3$, we can write $\Ec = \bigoplus_{a=1}^3\Oc_{\Cbb\Pbb^2}(k_a)$ for some triple of integers $k_1, k_2, k_2$.}
 The degree of $\Ec$ is $k = k_1 + k_2+k_3$ and the second exterior power is given by:
\[
\wedge^2\Ec = \Oc_{\Cbb\Pbb^2}(k_1+k_2) \oplus \Oc_{\Cbb\Pbb^2}(k_1+k_3)\oplus \Oc_{\Cbb\Pbb^2}(k_2+k_3).
\]
Now, the sheaf cohomology functor $H^i(-)$ on projective space commutes with (countably-many) direct sums, as discussed in \cite[p. 209]{HARTALG}, and so $h^i(\Ec) = \sum_{a=1}^3h^i(\Oc_{\Cbb\Pbb^2}(k_a))$. Then, in order to ensure \eqref{dmoeiciojciojc904jc409}, it suffices to choose the triple $(k_1, k_2, k_3)$ such that \eqref{splitvdcp3condn1}, \eqref{splitvdcp3condn2} and \eqref{splitvdcp3condn3} hold. This leads to the following constraints:
\begin{align}
k_1 + k_2>2
&&
k_1 + k_3 = -3
&&
\mbox{and}
&&
k_2 + k_3 < -3.
\label{jfuf484h84frifoe}
\end{align}
Evidently, a solution to \eqref{jfuf484h84frifoe} exists for any distinct pair of integers $(k_1, k_2)$ which satisfy $k_1 + k_2>2$. Then for such an $\Ec$ the conditions in \eqref{dmoeiciojciojc904jc409} will hold and so, by Lemma $\ref{dmoeiciojciojc904jc409dmoeiciojciojc904jc409}$, there will exist obstructed, second order thickenings $\Xfr_{(\Cbb\Pbb^2, E)}^{(2)}$ of $\Pi E^{(1)}$. 
\end{EX}

With regards to non-split bundles, it is not so straightforward to deduce the existence of obstructed thickenings as in the split case in Example \ref{8jf48f4unovu4no4oi}. In this article we will consider the next logical step after rank 3, split bundles being: rank 3, non-split, decomposable bundles.

\subsection{Non-split, Decomoposable Vector Bundles}
Non-split, decomposable vector bundles of any rank (greater than $3$) exist on the projective plane by virtue of the construction of indecomposable bundles (of any rank) by Schwarzenberger in \cite{SCH}. In rank 3 any non-split, decomposable bundle must necessarily be a direct sum of a rank 2, indecomposable bundle with a line bundle. Now suppose $E\ra \Cbb\Pbb^2$ has rank 3. Then $\wedge^3\Ec$ will be a line bundle and so $\wedge^3\Ec = \Oc_{\Cbb\Pbb^2}(k)$, for $k = \deg(E)$. We set $\Ec^\vee(k) := \Ec^\vee\otimes \Oc_{\Cbb\Pbb^2}(k)$. Our method of inferring the existence of obstructed thickenings here will follow that in Example \ref{8jf48f4unovu4no4oi}. That is, we will appeal to Lemma \ref{dmoeiciojciojc904jc409dmoeiciojciojc904jc409}. To this extent we present the following construction of a rank 2, indecomposable, holomorphic vector bundle from \cite{OSS}.

\begin{CON}\label{dnnriufj4jf904j094c}
Let $Y\subset \Cbb\Pbb^2$ comprise a collection of $m$ points, $m>0$, and let $\Jc_Y\subset \Cc_{\Cbb\Pbb^2}$ be the sub-sheaf of holomorphic functions on $\Cbb\Pbb^2$ which vanish on $Y$, i.e., an ideal sheaf. Then in \emph{\cite[p. 53]{OSS}} it is constructed a rank $2$, holomorphic vector bundle $F\ra\Cbb\Pbb^2$ and a global section $s\in H^0(\Cbb\Pbb^2, \Fc)$ 
such that firstly, for a fixed integer $k^\p<3$,
\begin{align}
c_1(F) = k^\p;&&c_2(F) = m
&&\mbox{and}&&
Y = \{\mbox{zeroes of $s$}\};
\label{dnieief894jf9j44j}
\end{align}
and secondly that the sheaf of holomorphic sections $\Fc$ of $F$ fits into the exact sequence:
\begin{align}
0 \ra \Cc_{\Cbb\Pbb^2}\stackrel{\cdot s}{\hookrightarrow} \Fc\twoheadrightarrow \Jc_Y(k^\p)\ra0
\label{deioieojf4jj43djd}
\end{align}
where $\Jc_Y(k^\p) = \Jc_Y\otimes\Oc_{\Cbb\Pbb^2}(k^\p)$ and $\Fc$ the sheaf of sections of $F$. We will be interested in the degree of $\Fc$ which, from \eqref{dnieief894jf9j44j} is $k^\p$ and so, henceforth, we will denote the bundle described here by $\Fc_{k^\p}$.
\end{CON}

More generally, a construction of rank 2, indecomposable bundles of a similar nature to that of $\Fc_{k^\p}$ in Construction \ref{dnnriufj4jf904j094c} on complex surfaces other than $\Cbb\Pbb^2$ is given in \cite[p. 726]{PAGGH}. We limit our considerations here to the projective plane. Integral to inferring the existence of some $E\ra \Cbb\Pbb^2$ such that \eqref{dmoeiciojciojc904jc409} holds is another famous theorem of Serre, which we state from \cite{OSS}:

\begin{THMSA}
Let $\Fc$ be a coherent, analytic sheaf on $\Cbb\Pbb^n$. Then there exists a $k_0\in \Zbb$ such that, for any $l\geq k_0$, the sheaf $\Fc(l)$ is generated by its global sections.\qed
\end{THMSA}

To elaborate on Serre's Theorem A, a sheaf $\Fc$ is said to be generated by its global sections if the evaluation map $H^0(\Fc) \otimes \Oc \ra \Fc$ is surjective, where $\Oc$ denotes the structure sheaf. In particular, if $\Fc$ is a holomorphic vector bundle on $\Cbb\Pbb^n$, we see that $h^0(\Fc(l))\neq 0$ for all $l\geq k_0$ and some $k_0\in \Zbb$. Now consider the bundle 
\[
\Ec_{(k^\p, l)} := \Fc_{k^\p}\oplus \Oc_{\Cbb\Pbb^2}(l), 
\]
for some $l$. Then $\Ec_{(k^\p, l)}$ will be the sheaf of holomorphic sections of a non-split, decomposable, rank 3 vector bundle $E_{(k^\p, l)}$ over $\Cbb\Pbb^2$. By construction we have that $\deg E_{(k^\p, l)} = k^\p + l$.

\begin{PROP}\label{hghgghdvvy3vdy3hdhhgwhfgabababbybwu}
For sufficiently small $l$ the bundle $\Ec_{(k^\p, l)}$ will be such that \eqref{dmoeiciojciojc904jc409} will be satisfied. 
\end{PROP}

\begin{proof}
We firstly have:
\begin{align}
\wedge^2\Ec_{(k^\p, l)} 
&\cong
\left(\wedge^2\Fc_{k^\p}\otimes \wedge^0\Oc_{\Cbb\Pbb^2}(l)\right)\oplus \left(\wedge^1\Fc_{k^\p}\otimes \wedge^1\Oc_{\Cbb\Pbb^2}(l)\right)
\notag
\\
&=
\Oc_{\Cbb\Pbb^2}(k^\p) \oplus \Fc_{k^\p}(l).
\label{040f903joi3mpwmvmep}
\end{align}
Now, the sheaf cohomology functor $H^i(\Cbb\Pbb^n, -)$ will commute with countably-many direct sums so we may deduce that $h^i(\Fc\oplus \Gc) = h^i(\Fc) + h^i(\Gc)$. Using this, Serre duality and \eqref{BOTT}, it will then follow from \eqref{040f903joi3mpwmvmep} that
\begin{align}
h^1(\Tfr_{\Cbb\Pbb^2}\otimes\wedge^2\Ec_{(k^\p, l)}) \geq h^1(\Om^1_{\Cbb\Pbb^2}(-k^\p-3)) = 1~~\mbox{iff}~k^\p = -3.
\label{fuihhfi4hf4h89h94}
\end{align}
Now, recall that we must have $k^\p< 3$ by construction of $\Fc_{k^\p}$ in Construction \ref{dnnriufj4jf904j094c}. In setting $k^\p = -3$, we will be assured in $h^1(\Tfr_{\Cbb\Pbb^2}\otimes\wedge^2\Ec_{(k^\p, l)})\neq0$ from \eqref{fuihhfi4hf4h89h94}. Now, as a result of setting $k^\p = -3$, note from \eqref{BOTT} that:
\begin{align}
h^2(\Tfr_{\Cbb\Pbb^2}\otimes\wedge^2\Ec_{(-3, l)})
&=
h^0(\Om^1_{\Cbb\Pbb^2} (0)) + h^0(\Om^1_{\Cbb\Pbb^2}\otimes \Fc^\vee_{-3}(-l-3))
\notag
\\
&= h^0(\Om^1_{\Cbb\Pbb^2}\otimes \Fc^\vee_{-3}(-l-3)).
\label{neiofio3f9830j303}
\end{align}
We are yet to impose any constraints on $l$ here. In appealing to Serre's Theorem A, we may choose $l$ sufficiently small (i.e., sufficiently negative) so that $-l-3 \gg 0$. This will ensure that $h^0(\Om^1_{\Cbb\Pbb^2}\otimes \Fc^\vee_{-3}(-l-3))\neq0$. Similarly, regarding the latter-most cohomology group in \eqref{dmoeiciojciojc904jc409}, we have:
\begin{align*}
h^2(\Ec^\vee(-3+l)) &= h^2(\Fc^\vee_{-3}(-3+l)) + h^2(\Oc_{\Cbb\Pbb^2}(-3))
= h^0(\Fc_{-3}(-l+3)) 
\end{align*}
and just as in \eqref{neiofio3f9830j303} we see, for sufficiently small $l$, that $h^0(\Fc_{-3}(-l+3))>0$ by Serre's Theorem A. The proposition now follows.
\end{proof}

From Lemma \ref{dmoeiciojciojc904jc409dmoeiciojciojc904jc409} and Proposition \ref{hghgghdvvy3vdy3hdhhgwhfgabababbybwu} we conclude:

\begin{THM}
Let $E_{(-3, l)} \ra \Cbb\Pbb^2$ be the vector bundle whose sheaf of sections is $\Ec_{(-3, l)}$. Then for sufficiently small $l$, there will exist an obstructed, second order thickening of $\Pi E^{(1)}_{(-3, l)}$. \qed
\end{THM}

\section{Concluding Remarks}

\noindent
A common technique in studies of supersymmetric sigma-models and field theories is to move from the superspace formulation to the component formulation, which is achieved by carrying out a Berezin integral at some stage. See \cite{FREEDSUSY, DELFREEDSUSY} for details of this general procedure for superspace Lagrangians. Importantly, employing this technique allows for well studied methods from geometry to become readily applicable to calculate quantities of interest. There are however instances in which this technique will not be well-defined and this issue, of well-definedness, can be directly related to the subtleties of obstruction theory. A motivating example highlighting precisely this issue is superstring theory---the quantities of interest here being scattering amplitudes for the superstring. For more on this topic see \cite{HOKERPHONG1, HOKER1}. It is a hope of the author that obstruction theory can provide insights into supersymmetric theories more generally. For instance, as a means to study certain supersymmetric field theories within the superspace formulation itself. At the very least, it seems clear that to circumvent some of the issues plaguing superstring theory, further developments in complex supergeometry will be desirable and these will certainly require a greater understanding of obstruction theory.

\bibliographystyle{alpha}
\bibliography{Bibliography}

\newcommand{\etalchar}[1]{$^{#1}$}
\begin{thebibliography}{NCP{\etalchar{+}}17}

\bibitem[Bat79]{BAT}
M.~Batchelor.
\newblock The structure of supermanifolds.
\newblock {\em Trans. Amer. Math. Soc.}, (253):329--338, 1979.

\bibitem[Ber87]{BER}
F.~A. Berezin.
\newblock {\em Introduction to Superanalysis}.
\newblock D. Reidel Publishing Company, 1987.

\bibitem[Bet16]{BETTPHD}
K.~Bettadapura.
\newblock {\em Obstruction Theory for Supermanifolds and Deformations of
  Superconformal Structures}.
\newblock PhD thesis, The Australian National University,
  \href{http://hdl.handle.net/1885/110239}{hdl.handle.net/1885/110239},
  December 2016.

\bibitem[Bet18a]{BETTEMB}
K.~Bettadapura.
\newblock Embeddings of complex supermanifolds.
\newblock Available at:
  \href{http://arxiv.org/abs/1806.02763}{arXiv:1806.02763} [math.AG], 2018.

\bibitem[Bet18b]{BETTHIGHOBS}
K.~Bettadapura.
\newblock Higher obstructions of complex supermanifolds.
\newblock {\em SIGMA}, 14(094), 2018.

\bibitem[Bet18c]{BETTVAR}
K.~Bettadapura.
\newblock Projective superspace varieties, superspace quadrics and
  non-splitting.
\newblock Available at:
  \href{http://arxiv.org/abs/1810.10200}{arXiv:1810.10200} [math.AG], 2018.

\bibitem[Bry08]{BRY}
J.~Brylinski.
\newblock {\em Loop Spaces, Characteristic Classes and Geometric Quantization}.
\newblock Modern Birkhaeuser Classics, 2008.

\bibitem[CNR17]{NOJACY}
S.~L. Cacciatori, S.~Noja, and R.~Re.
\newblock Non-projected {C}alabi-{Y}au supermanifolds over {$\mathbb P^2$}.
\newblock Available at:
  \href{http://arxiv.org/abs/1706:01354}{arXiv:1706:01354} [math.AG], 2017.

\bibitem[DF99]{DELFREEDSUSY}
P.~Deligne and D.~S. Freed.
\newblock {\em Quantum Fields and Strings: A course for Mathematicians},
  volume~1, chapter Supersolutions, pages 227--355.
\newblock American Mathematical Society and Institute for Advanced Studies,
  Princeton, 1999.

\bibitem[D'H14]{HOKER1}
E.~D'Hoker.
\newblock Topics in two-loop superstring perturbation theory.
\newblock Available: \href{http://arxiv.org/abs/1403.5494}{arXiv:1403.5494}
  [hep-th], 2014.

\bibitem[DM99]{QFAS}
P.~Deligne and J.~W. Morgan.
\newblock {\em Quantum Fields and Strings: A course for Mathematicians},
  volume~1, chapter Notes on Supersymmetry (following Joseph Bernstein), pages
  41--97.
\newblock American Mathematical Society, Providence, 1999.

\bibitem[DP02]{HOKERPHONG1}
E.~D'Hoker and D.~Phong.
\newblock Lectures on two loop superstrings.
\newblock {\em Conf. Proc. C0208124}, pages 85--123, 2002.

\bibitem[DW14]{DW2}
R.~Donagi and E.~Witten.
\newblock Super {Atiyah} classes and obstructions to splitting of supermoduli
  space.
\newblock available at: \href{http://arxiv.org/abs/1404.6257}{arXiv:1404.6257}
  [hep-th], 2014.

\bibitem[DW15]{DW1}
R.~Donagi and E.~Witten.
\newblock Supermoduli space is not projected.
\newblock In {\em Proc. Symp. Pure Math.}, volume~90, pages 19--72, 2015.

\bibitem[EL86]{EASTBRU}
M.~Eastwood and C.~LeBrun.
\newblock Thickening and supersymmetric extensions of complex manifolds.
\newblock {\em Am. J. Math}, 108(5):1177--1192, 1986.

\bibitem[Fre99]{FREEDSUSY}
D.~S. Freed.
\newblock {\em Five Lectures on Supersymmetry}.
\newblock American Mathematical Society, 1999.

\bibitem[GH78]{PAGGH}
Phillip Griffiths and Joe Harris.
\newblock {\em Principles of Algebraic Geometry}.
\newblock John Wiley and Sons, 1978.

\bibitem[Gre82]{GREEN}
P.~Green.
\newblock On holomorphic graded manifolds.
\newblock {\em Proc. Amer. Math. Soc.}, 85(4):587--590, 1982.

\bibitem[Gri66]{GRIFF}
P.~Griffiths.
\newblock The extension problem in complex analysis {I}{I}: embeddings with
  positive normal bundle.
\newblock {\em Am. J. Math}, (88):366--446, 1966.

\bibitem[Gro55]{GROTHNONAB}
A.~Grothendieck.
\newblock A general theory of fibre spaces with structure sheaf.
\newblock {\em Univ. of Kansas}, 1955.

\bibitem[Har77]{HARTALG}
R.~Hartshorne.
\newblock {\em Algebraic Geometry}.
\newblock Springer, 1977.

\bibitem[Har10]{HARTDEF}
R.~Hartshorne.
\newblock {\em Deformation Theory}.
\newblock Springer, 2010.

\bibitem[Kap15]{KAPSUSY}
M.~Kapranov.
\newblock Supergeometry in mathematics and physics.
\newblock available at:
  {\href{http://arxiv.org/abs/1512.07042}{arXiv:1512.07042}} [math.AG], 2015.

\bibitem[Kod86]{KS}
K.~Kodaira.
\newblock {\em Complex Manifolds and Deformation of Complex Structures}.
\newblock Springer, 1986.

\bibitem[Lee06]{LEE}
J.~M. Lee.
\newblock {\em Introduction to Smooth Manifolds}.
\newblock Springer-Verlag, 2006.

\bibitem[Lei74]{LEITSPEC}
D.~Leites.
\newblock Spectra of graded-commutative rings.
\newblock {\em Uspekhi. Mat. Nauk}, 29(3):209--210, 1974.

\bibitem[Lei80]{LEI}
D.~Leites.
\newblock Introduction to the theory of supermanifolds.
\newblock {\em Russian Math. Surveys}, 35(1):1--64, 1980.

\bibitem[LPW90]{LEBRUN}
C.~LeBrun, Y.~Poon, and R.~Wells.
\newblock Projective embeddings of complex supermanifolds.
\newblock {\em Comm. Math. Phys.}, 126(3):433--452, 1990.

\bibitem[Man88]{YMAN}
Y.~Manin.
\newblock {\em Gauge Fields and Complex Geometry}.
\newblock Springer-Verlag, 1988.

\bibitem[NCP{\etalchar{+}}17]{NOJAONESCY}
S~Noja, S.L. Cacciatori, F.D. Piazza, et~al.
\newblock One-dimensional super calabi-yau manifolds and their mirrors.
\newblock {\em J. High Energ. Phys.}, 4(094), 2017.

\bibitem[Noj18a]{NOJAEMB}
S.~Noja.
\newblock Non-projected supermanifolds and embeddings in super grassmannians.
\newblock Available at:
  \href{http://arxiv.org/abs/1808.09817}{arXiv:1808.09817} [math.AG], 2018.

\bibitem[Noj18b]{NOJAPI}
S.~Noja.
\newblock Supergeometry of {$\Pi$}-{P}rojective spaces.
\newblock {\em J. Geom. and Phys.}, 124:286--299, 2018.

\bibitem[Oni97]{ONISHMOD}
A.~L. Onishchik.
\newblock A moduli problem related to complex supermanifolds.
\newblock In Y.~Khakimdjanov, Y.~Goze, and M.~Ayupov, editors, {\em Algebra and
  Operator Theory: Proceedings of the Colloquium in Tashkent}, pages 13--24,
  1997.

\bibitem[Oni98]{ONISHNS}
A.~L. Onishchik.
\newblock A construction of non-split supermanifolds.
\newblock {\em Ann. Glob. Anal. Geom.}, (16):309--333, 1998.

\bibitem[Oni99]{ONISHCLASS}
A.~L. Onishchik.
\newblock On the classification of complex analytic supermanifolds.
\newblock {\em Lobachevskii J. Math.}, pages 47--70, 1999.

\bibitem[Oni00]{ONISHNSCOT}
A.~L. Onishchik.
\newblock Non-split supermanifolds associated with the cotangent bundle.
\newblock In {\em Lie Groups, Geometric Structures and Differential Equations,
  One Hundred Years After Sophus Lie}, number 1150, pages 45--53.
  Surikaisekikenkyusho Kokyuroku, 2000.

\bibitem[OSS10]{OSS}
C.~Okonek, M.~Schneider, and H.~Spindler.
\newblock {\em Vector Bundles on Complex Projective Spaces}.
\newblock Modern Birkhaeuser Classics, 2010.

\bibitem[RC85]{RABGLOB}
J.~Rabin and L.~Crane.
\newblock How different are the supermanifolds of {Rogers} and {DeWitt}?
\newblock {\em Comm. Math. Phys.}, 102:123--137, 1985.

\bibitem[Rog07]{ROG}
A.~Rogers.
\newblock {\em Supermanifolds: Theory and Applications}.
\newblock World Scientific, 2007.

\bibitem[Rot85]{ROTH}
M.~Rothstein.
\newblock Deformations of complex supermanifolds.
\newblock {\em Proc. Amer. Math. Soc.}, 95(2):255--260, October 1985.

\bibitem[Sch61]{SCH}
R.~L.~E. Schwarzenberger.
\newblock Vector bundles on the projective plane.
\newblock {\em Proc. London Math. Soc.}, 11:623--640, 1961.

\bibitem[{Sta}16]{STACKPROJ}
{Stacks Project Authors}.
\newblock Stacks {Project}.
\newblock {\url{http://stacks.math.columbia.edu}}, 2016.

\bibitem[Vai90]{VAIN}
Y.~Vaintrob.
\newblock Deformation of complex superspaces and coherent sheaves on them.
\newblock {\em J. Soviet Math.}, 51(1):2140--2188, August 1990.

\end{thebibliography}

\hfill
\\
\noindent
\small
\textsc{
Kowshik Bettadapura, 
\\\\
Mathematical Sciences Institute, Australian National University, Canberra, ACT 2601, Australia
\\\\
Yau Mathematical Sciences Center, Tsinghua University, Beijing, Haidian, 100084, China}
\\\\
\emph{E-mail address:} \href{mailto:kowshik@mail.tsinghua.edu.cn}{kowshik@mail.tsinghua.edu.cn}

\end{document}